\RequirePackage[l2tabu,orthodox]{nag}
\documentclass[submission,copyright,creativecomons,noderivs,10pt]{eptcs}

\providecommand{\event}{QPL 2018}

\pdfoutput=1

\usepackage[T1]{fontenc}
\usepackage{fixltx2e}
\usepackage{microtype} 
\usepackage{xspace}
\usepackage{enumitem}

\newcommand{\figureline}{\rule{\textwidth}{0.5pt}}

\makeatletter
\newcommand\etc{etc\@ifnextchar.{}{.\@}\xspace}
\newcommand\ie{i.e.\@\xspace}  

\makeatother

\newcommand{\vstrut}[2][0pt]{\rule[#1]{0pt}{#2}}

\usepackage{graphicx}

\newcommand{\inlinegraphic}[2]{
  \dimendef\grafheight=255\dimendef\grafvshift=254
  \grafheight=#1
  \grafvshift=-0.5\grafheight
  \advance\grafvshift by 0.5ex
  \raisebox{\grafvshift}{\includegraphics[height=\grafheight]{images/#2}\xspace}
}

\newcommand{\ninlinegraphic}[2][1.0]{
  \dimendef\grafheight=255\dimendef\grafvshift=254
  \setbox0 = \hbox{\scalebox{#1}{\includegraphics{images/#2}}}
  \grafheight=\the\ht0
  \grafvshift=-0.5\grafheight
  \advance\grafvshift by 0.5ex
  \raisebox{\grafvshift}{\includegraphics[height=\grafheight]{images/#2}\xspace}
}

\newcommand{\inline}[1]{
  \raisebox{0.5ex}{\;#1\;}
}

\usepackage{latexsym}
\usepackage{amssymb}
\usepackage{amsmath} 
\usepackage{stmaryrd}
\usepackage{mathtools}

\usepackage{amsthm}
\newtheorem{theorem}{Theorem}[section]
\newtheorem{proposition}[theorem]{Proposition}
\newtheorem{lemma}[theorem]{Lemma}

\theoremstyle{definition}
\theoremstyle{definition}
\theoremstyle{definition}\newtheorem{definition}[theorem]{Definition}
\theoremstyle{definition}
\theoremstyle{definition}\newtheorem{remark}[theorem]{Remark}
\theoremstyle{definition}

\newcommand{\denote}[1]{
\llbracket #1 \rrbracket} 
\newcommand{\ldenote}[1]{\left\llbracket #1 \right\rrbracket}

\newcommand{\sizeof}[1]{
  \left|#1\right|}

\newcommand{\ket}[1]{
    \ensuremath{\left|  #1 \right\rangle}\xspace}

\newcommand{\CZ}{\ensuremath{\wedge Z}\xspace}

\ifx\numericids\undefined
\newcommand{\id}[1]{\ensuremath{\mathrm{id}_{#1}}}
\else
\newcommand{\id}[1]{\ensuremath{1_{#1}}}
\fi

\newcommand{\fdhilb}{
\ensuremath{\mathbf{fdHilb}}\xspace}

\newcommand{\zxcalculus}{\textsc{zx}-calculus\xspace}
\newcommand{\zxdiagram}{\textsc{zx}-diagram\xspace}
\newcommand{\zxdiagrams}{\textsc{zx}-diagrams\xspace}
\newcommand{\ZX}{\ensuremath{\mathbf{ZX}}\xspace}
\newcommand{\ZXs}{\ensuremath{\mathbf{ZX_s}}\xspace}

\usepackage{dsfont}

\newcommand{\Cliff}{\ensuremath{\mathbf{Cliff}}\xspace}

\newcommand{\C}[1]{%
\ensuremath{\mathcal{C}_{#1}\xspace}}

\newcommand{\CC}[1]{%
\ensuremath{\mathfrak{C}_{#1}\xspace}}

\renewcommand{\P}[1]{%
\ensuremath{\mathcal{P}_{#1}\xspace}}

\newcommand{\eq}{\ensuremath{\stackrel{*}{\leftrightarrow}}}

\newcommand{\rw}{\ensuremath{\stackrel{*}{\rightarrow}}}

\newcommand{\CNOT}{\ensuremath{\textsc{cnot}}\xspace}
\newcommand{\TONC}{\ensuremath{\textsc{tonc}}\xspace}

\usepackage{hyperref}

\usepackage{tikzit}  
\newcommand{\inlinetikzfig}[1]{%
  \raisebox{0.5ex}{{\;\input{#1.tikz}}}}

\pgfdeclarelayer{background} 
\pgfdeclarelayer{foreground}
\pgfdeclarelayer{edgelayer}
\pgfdeclarelayer{nodelayer}
\pgfsetlayers{background,edgelayer,nodelayer,main,foreground} 

\tikzstyle{halfsize}=[x=0.5cm, y=0.5cm]
\tikzstyle{normalsize}=[]
\tikzstyle{doublesize}=[]

\tikzstyle{(null)}=[]
\tikzstyle{plain}=[]

\usepackage{tikzzx}
\usepackage{tikzcircuits}

\newcommand{\rpp}{%
  \begin{tikzpicture}[quanto]
    \node [rpp] (c) at (-2.0,1.175) {};
    \node [boundary vertex] (b) at (-3.0,1.175) {};
    \node [boundary vertex] (a) at (-1.0,1.175) {};
    \draw [] (a) to (c);
    \draw [] (c) to (b);
  \end{tikzpicture}
}

\newcommand{\rpi}{%
  \begin{tikzpicture}[quanto]
    \node [rpi] (c) at (-2.0,1.175) {};
    \node [boundary vertex] (b) at (-3.0,1.175) {};
    \node [boundary vertex] (a) at (-1.0,1.175) {};
    \draw [] (a) to (c);
    \draw [] (c) to (b);
  \end{tikzpicture}
}

\newcommand{\rmm}{%
  \begin{tikzpicture}[quanto]
    \node [rmm] (c) at (-2.0,1.175) {};
    \node [boundary vertex] (b) at (-3.0,1.175) {};
    \node [boundary vertex] (a) at (-1.0,1.175) {};
    \draw [] (a) to (c);
    \draw [] (c) to (b);
  \end{tikzpicture}
}

\newcommand{\gpp}{%
  \begin{tikzpicture}[quanto]
    \node [gpp] (c) at (-2.0,1.175) {};
    \node [boundary vertex] (b) at (-3.0,1.175) {};
    \node [boundary vertex] (a) at (-1.0,1.175) {};
    \draw [] (a) to (c);
    \draw [] (c) to (b);
  \end{tikzpicture}
}

\newcommand{\gpi}{%
  \begin{tikzpicture}[quanto]
    \node [gpi] (c) at (-2.0,1.175) {};
    \node [boundary vertex] (b) at (-3.0,1.175) {};
    \node [boundary vertex] (a) at (-1.0,1.175) {};
    \draw [] (a) to (c);
    \draw [] (c) to (b);
  \end{tikzpicture}
}

\newcommand{\gmm}{%
  \begin{tikzpicture}[quanto]
    \node [gmm] (c) at (-2.0,1.175) {};
    \node [boundary vertex] (b) at (-3.0,1.175) {};
    \node [boundary vertex] (a) at (-1.0,1.175) {};
    \draw [] (a) to (c);
    \draw [] (c) to (b);
  \end{tikzpicture}
}

\newcommand{\SWAP}{%
  \begin{tikzpicture}[quanto,halfsize]
    \node [boundary vertex] (tl) at (1.0,1.0) {};
    \node [boundary vertex] (br) at (-2,3.0) {};
    \node [boundary vertex] (bl) at (-2,1.0) {};
    \node [boundary vertex] (tr) at (1.0,3.0) {};
    \draw [out=0,in=180] (br) to (tl) ;
    \draw [out=0,in=180] (bl) to (tr) ;
  \end{tikzpicture}
}

\newcommand{\circbox}[1]{\inline{%
\begin{tikzpicture}[circuit]
  \node  (0) at (1,0) {};
  \node [box] (1) at (0, 0) {$#1$};
  \node  (2) at (-1,0) {};
  \draw  (0) to (1);
  \draw  (1) to (2);
\end{tikzpicture}
}}

\newcommand{\circX}[1]{\circbox{X_#1}}
\newcommand{\circZ}[1]{\circbox{Z_#1}}
\newcommand{\circH}{\circbox{H_{}}}
\newcommand{\circS}{\circbox{S_{}}}
\newcommand{\circV}{\circbox{V_{}}}

\newcommand{\circCX}{\inline{%
\begin{tikzpicture}[circuit]
  \node  (0) at (1, 0) {};
  \node  (1) at (1, 1) {};
  \circcnot{3}{0,1}{2}{0,0}
  \node  (4) at (-1, 0) {};
  \node  (5) at (-1,1) {};
  \draw  (1) to (3);
  \draw  (0) to (2);
  \draw  (3) to (5);
  \draw  (2) to (4);
\end{tikzpicture}
}}

\newcommand{\HcliffI}{\inline{%
\begin{tikzpicture}[quanto,halfsize]
  \node [boundary vertex] (0) at (2.0,0) {};
  \node [gpp] (1) at (1,0) {};
  \node [rpp] (2) at (0,0) {};
  \node [gpp] (3) at (-1,0) {};
  \node [boundary vertex] (4) at (-2.0,0) {};
  \draw  (0) to (1);
  \draw  (1) to (2);
  \draw  (2) to (3);
  \draw  (4) to (3);
\end{tikzpicture}
}}

\newcommand{\HcliffII}{\inline{%
\begin{tikzpicture}[quanto,halfsize]
  \node [boundary vertex] (0) at (2.0,0) {};
  \node [rmm] (1) at (1,0) {};
  \node [gmm] (2) at (0,0) {};
  \node [rmm] (3) at (-1,0) {};
  \node [boundary vertex] (4) at (-2.0,0) {};
  \draw  (0) to (1);
  \draw  (1) to (2);
  \draw  (2) to (3);
  \draw  (4) to (3);
\end{tikzpicture}
}}

\newcommand{\HcliffIII}{\inline{%
\begin{tikzpicture}[quanto,halfsize]
  \node [boundary vertex] (0) at (2.0,0) {};
  \node [rpp] (1) at (1,0) {};
  \node [rpi] (2) at (0,0) {};
  \node [gmm] (3) at (-1,0) {};
  \node [rmm] (4) at (-2,0) {};
  \node [boundary vertex] (5) at (-3.0,0) {};
  \draw  (0) to (1);
  \draw  (1) to (2);
  \draw  (2) to (3);
  \draw  (4) to (3);
  \draw  (4) to (5);
\end{tikzpicture}
}}

\newcommand{\HcliffIV}{\inline{%
\begin{tikzpicture}[quanto,halfsize]
  \node [boundary vertex] (0) at (2.0,0) {};
  \node [rpp] (1) at (1,0) {};
  \node [gpp] (2) at (0,0) {};
  \node [rpi] (3) at (-1,0) {};
  \node [rmm] (4) at (-2,0) {};
  \node [boundary vertex] (5) at (-3.0,0) {};
  \draw  (0) to (1);
  \draw  (1) to (2);
  \draw  (2) to (3);
  \draw  (4) to (3);
  \draw  (4) to (5);
\end{tikzpicture}
}}

\newcommand{\HcliffV}{\inline{%
\begin{tikzpicture}[quanto,halfsize]
  \node [boundary vertex] (0) at (2.0,0) {};
  \node [rpp] (1) at (1,0) {};
  \node [gpp] (2) at (0,0) {};
  \node [rpp] (3) at (-1,0) {};
  \node [boundary vertex] (5) at (-2.0,0) {};
  \draw  (0) to (1);
  \draw  (1) to (2);
  \draw  (2) to (3);
  \draw  (3) to (5);
\end{tikzpicture}
}}

\newcommand{\HcliffVI}{\inline{%
\begin{tikzpicture}[quanto,halfsize]
  \node [boundary vertex] (0) at (2.0,0) {};
  \node [gmm] (1) at (1,0) {};
  \node [rmm] (2) at (0,0) {};
  \node [gmm] (3) at (-1,0) {};
  \node [boundary vertex] (5) at (-2.0,0) {};
  \draw  (0) to (1);
  \draw  (1) to (2);
  \draw  (2) to (3);
  \draw  (3) to (5);
\end{tikzpicture}
}}

\newcommand{\HcliffIb}{\inline{%
\begin{tikzpicture}[quanto,halfsize]
  \node [boundary vertex] (0) at (2.0,0) {};
  \node [gmm] (1) at (1,0) {};
  \node [rmm] (2) at (0,0) {}; 
  \node [gpp] (3) at (-1,0) {};
  \node [boundary vertex] (4) at (-2,0) {};
  \draw (0) to (1);
  \draw (1) to (2);
  \draw (2) to (3);
  \draw (3) to (4);
\end{tikzpicture}
}}

\newcommand{\HcliffIIb}{\inline{%
\begin{tikzpicture}[quanto,halfsize]
  \node [boundary vertex] (0) at (2.0,0) {};
  \node [rpp] (1) at (1,0) {};
  \node [gpp] (2) at (0,0) {}; 
  \node [rmm] (3) at (-1,0) {};
  \node [boundary vertex] (4) at (-2,0) {};
  \draw (0) to (1);
  \draw (1) to (2);
  \draw (2) to (3);
  \draw (3) to (4);
\end{tikzpicture}
}}
\newcommand{\HcliffIIIb}{\inline{%
\begin{tikzpicture}[quanto,halfsize]
  \node [boundary vertex] (0) at (2.0,0) {};
  \node [gpp] (1) at (1,0) {};
  \node [rpp] (2) at (0,0) {}; 
  \node [gmm] (3) at (-1,0) {};
  \node [boundary vertex] (4) at (-2,0) {};
  \draw (0) to (1);
  \draw (1) to (2);
  \draw (2) to (3);
  \draw (3) to (4);
\end{tikzpicture}
}}
\newcommand{\HcliffIVb}{\inline{%
\begin{tikzpicture}[quanto,halfsize]
  \node [boundary vertex] (0) at (2.0,0) {};
  \node [rmm] (1) at (1,0) {};
  \node [gmm] (2) at (0,0) {}; 
  \node [rpp] (3) at (-1,0) {};
  \node [boundary vertex] (4) at (-2,0) {};
  \draw (0) to (1);
  \draw (1) to (2);
  \draw (2) to (3);
  \draw (3) to (4);
\end{tikzpicture}
}}

\newcommand{\HcliffIc}{\inline{%
\begin{tikzpicture}[quanto,halfsize]
  \node [boundary vertex] (0) at (2.0,0) {};
  \node [gpp] (1) at (1,0) {};
  \node [rmm] (2) at (0,0) {}; 
  \node [gpp] (3) at (-1,0) {};
  \node [boundary vertex] (4) at (-2,0) {};
  \draw (0) to (1);
  \draw (1) to (2);
  \draw (2) to (3);
  \draw (3) to (4);
\end{tikzpicture}
}}

\newcommand{\HcliffIIc}{\inline{%
\begin{tikzpicture}[quanto,halfsize]
  \node [boundary vertex] (0) at (2.0,0) {};
  \node [rmm] (1) at (1,0) {};
  \node [gpp] (2) at (0,0) {}; 
  \node [rmm] (3) at (-1,0) {};
  \node [boundary vertex] (4) at (-2,0) {};
  \draw (0) to (1);
  \draw (1) to (2);
  \draw (2) to (3);
  \draw (3) to (4);
\end{tikzpicture}
}}

\newcommand{\HcliffIIIc}{\inline{%
\begin{tikzpicture}[quanto,halfsize]
  \node [boundary vertex] (0) at (2.0,0) {};
  \node [gmm] (1) at (1,0) {};
  \node [rpp] (2) at (0,0) {}; 
  \node [gmm] (3) at (-1,0) {};
  \node [boundary vertex] (4) at (-2,0) {};
  \draw (0) to (1);
  \draw (1) to (2);
  \draw (2) to (3);
  \draw (3) to (4);
\end{tikzpicture}
}}

\newcommand{\HcliffIVc}{\inline{%
\begin{tikzpicture}[quanto,halfsize]
  \node [boundary vertex] (0) at (2.0,0) {};
  \node [rpp] (1) at (1,0) {};
  \node [gmm] (2) at (0,0) {}; 
  \node [rpp] (3) at (-1,0) {};
  \node [boundary vertex] (4) at (-2,0) {};
  \draw (0) to (1);
  \draw (1) to (2);
  \draw (2) to (3);
  \draw (3) to (4);
\end{tikzpicture}
}}

\newcommand{\HcliffId}{\inline{%
\begin{tikzpicture}[quanto,halfsize]
  \node [boundary vertex] (0) at (2.0,0) {};
  \node [gmm] (1) at (1,0) {};
  \node [rpp] (2) at (0,0) {}; 
  \node [gpp] (3) at (-1,0) {};
  \node [boundary vertex] (4) at (-2,0) {};
  \draw (0) to (1);
  \draw (1) to (2);
  \draw (2) to (3);
  \draw (3) to (4);
\end{tikzpicture}
}}

\newcommand{\HcliffIId}{\inline{%
\begin{tikzpicture}[quanto,halfsize]
  \node [boundary vertex] (0) at (2.0,0) {};
  \node [rpp] (1) at (1,0) {};
  \node [gmm] (2) at (0,0) {}; 
  \node [rmm] (3) at (-1,0) {};
  \node [boundary vertex] (4) at (-2,0) {};
  \draw (0) to (1);
  \draw (1) to (2);
  \draw (2) to (3);
  \draw (3) to (4);
\end{tikzpicture}
}}

\newcommand{\HcliffIIId}{\inline{%
\begin{tikzpicture}[quanto,halfsize]
  \node [boundary vertex] (0) at (2.0,0) {};
  \node [gpp] (1) at (1,0) {};
  \node [rmm] (2) at (0,0) {}; 
  \node [gmm] (3) at (-1,0) {};
  \node [boundary vertex] (4) at (-2,0) {};
  \draw (0) to (1);
  \draw (1) to (2);
  \draw (2) to (3);
  \draw (3) to (4);
\end{tikzpicture}
}}

\newcommand{\HcliffIVd}{\inline{%
\begin{tikzpicture}[quanto,halfsize]
  \node [boundary vertex] (0) at (2.0,0) {};
  \node [rmm] (1) at (1,0) {};
  \node [gpp] (2) at (0,0) {}; 
  \node [rpp] (3) at (-1,0) {};
  \node [boundary vertex] (4) at (-2,0) {};
  \draw (0) to (1);
  \draw (1) to (2);
  \draw (2) to (3);
  \draw (3) to (4);
\end{tikzpicture}
}}

\newcommand{\toncHHHH}{%
  \begin{tikzpicture}[quanto,halfsize]
    \node [boundary vertex] (a) at (0,-1.0) {};
    \node [boundary vertex] (d) at (-4,-3.0) {};
    \node [boundary vertex] (c) at (-4,-1.0) {};
    \node [boundary vertex] (b) at (0,-3.0) {};
    \node [green smallvertex] (f) at (-2.0,-1.0) {};
    \node [red smallvertex] (e) at (-2.0,-3.0) {};
    \draw [] (f) to (e);
    \draw [] (e) to (d);
    \draw [] (f) to (c);
    \draw [] (b) to (e);
    \draw [] (a) to (f);
    \node [hadamard vertex] (h1) at (-1.0,-1.0) {};
    \node [hadamard vertex] (h2) at (-1.0,-3.0) {};
    \node [hadamard vertex] (h3) at (-3.0,-1.0) {};
    \node [hadamard vertex] (h4) at (-3.0,-3.0) {};
  \end{tikzpicture}
}

\newcommand{\toncNF}{%
  \begin{tikzpicture}[style=quanto,halfsize]
    \node [style=boundary vertex] (a) at (-1.00, -1.00) {};
    \node [style=boundary vertex] (d) at (5.00, -3.00) {};
    \node [style=boundary vertex] (c) at (5.00, -1.00) {};
    \node [style=boundary vertex] (b) at (-1.00, -3.00) {};
    \node [style=green smallvertex] (f) at (2.00, -1.00) {};
    \node [style=red smallvertex] (e) at (2.00, -3.00) {};
    \draw [style=] (f) to (e);
    \draw [style=] (e) to (d);
    \draw [style=] (f) to (c);
    \draw [style=] (b) to (e);
    \draw [style=] (a) to (f);
    \node [style=gmm] (gmmL) at (0.00, -1.00) {};
    \node [style=rmm] (rmmL) at (1.00, -1.00) {};
    \node [style=rmm] (rmmR) at (0.00, -3.00) {};
    \node [style=gmm] (gmmR) at (1.00, -3.00) {};

    \node [style=rpp] (rppL) at (3.00, -1.00) {};
    \node [style=gpp] (gppL) at (4.00, -1.00) {};
    \node [style=gpp] (gppR) at (3.00, -3.00) {};
    \node [style=rpp] (rppR) at (4.00, -3.00) {};
  \end{tikzpicture}
}

\newcommand{\CNOTCNOTsCaseAi}{%
\begin{tikzpicture}[quanto,halfsize]
  \begin{pgfonlayer}{nodelayer}
    \node [style=boundary vertex] (0) at (-2.00, 1.00) {};
    \node [style=boundary vertex] (1) at (-2.00, -1.00) {};
    \node [style=green smallvertex] (2) at (-1.00, 1.00) {};
    \node [style=red smallvertex] (3) at (-1.00, -1.00) {};
    \node [style=rpp] (4) at (0.00, 1.00) {};
    \node [style=green smallvertex] (5) at (1.00, 1.00) {};
    \node [style=red smallvertex] (6) at (1.00, -1.00) {};
    \node [style=boundary vertex] (7) at (2.00, 1.00) {};
    \node [style=boundary vertex] (8) at (2.00, -1.00) {};
  \end{pgfonlayer}
  \begin{pgfonlayer}{edgelayer}
    \draw (2) to (3);
    \draw (0) to (7);
    \draw (5) to (6);
    \draw (1) to (8);
  \end{pgfonlayer}
\end{tikzpicture}}

\newcommand{\CNOTCNOTsCaseAii}{%
\begin{tikzpicture}[quanto,halfsize]
  \begin{pgfonlayer}{nodelayer}
    \node [style=boundary vertex] (0) at (-2.00, 1.00) {};
    \node [style=boundary vertex] (1) at (-2.00, -0.75) {};
    \node [style=green smallvertex] (2) at (-1.00, 1.00) {};
    \node [style=rpp] (3) at (0.00, 2.75) {};
    \node [style=red smallvertex] (4) at (0.00, 2.00) {};
    \node [style=red smallvertex] (5) at (0.00, 0.00) {};
    \node [style=red smallvertex] (6) at (0.00, -0.75) {};
    \node [style=green smallvertex] (7) at (1.00, 1.00) {};
    \node [style=boundary vertex] (8) at (2.00, 1.00) {};
    \node [style=boundary vertex] (9) at (2.00, -0.75) {};
  \end{pgfonlayer}
  \begin{pgfonlayer}{edgelayer}
    \draw (2) to (4);
    \draw (4) to (7);
    \draw (7) to (8);
    \draw (2) to (5);
    \draw (7) to (5);
    \draw (1) to (9);
    \draw (5) to (6);
    \draw (0) to (2);
    \draw (4) to (3);
  \end{pgfonlayer}
\end{tikzpicture}}

\newcommand{\CNOTCNOTsCaseAiii}{%
\begin{tikzpicture}[quanto,halfsize]
  \begin{pgfonlayer}{nodelayer}
    \node [style=boundary vertex] (0) at (-2.00, 1.00) {};
    \node [style=boundary vertex] (1) at (-2.00, -1.00) {};
    \node [style=red smallvertex] (2) at (0.00, 1.00) {};
    \node [style=green smallvertex] (3) at (0.00, 0.00) {};
    \node [style=red smallvertex] (4) at (0.00, -1.00) {};
    \node [style=rpp] (5) at (1.00, 0.00) {};
    \node [style=boundary vertex] (6) at (2.00, 1.00) {};
    \node [style=boundary vertex] (7) at (2.00, -1.00) {};
  \end{pgfonlayer}
  \begin{pgfonlayer}{edgelayer}
    \draw (5) to (3);
    \draw (2) to (6);
    \draw (3) to (4);
    \draw (1) to (7);
    \draw (2) to (3);
    \draw (0) to (2);
  \end{pgfonlayer}
\end{tikzpicture}}

\newcommand{\CNOTCNOTsCaseAiv}{%
\begin{tikzpicture}[quanto,halfsize]
  \begin{pgfonlayer}{nodelayer}
    \node [style=boundary vertex] (0) at (-2.25, 1.00) {};
    \node [style=boundary vertex] (1) at (-2.25, -1.50) {};
    \node [style=gpp] (2) at (-1.50, 1.00) {};
    \node [style=gmm] (3) at (-0.75, 1.00) {};
    \node [style=red smallvertex] (4) at (0.00, 1.00) {};
    \node [style=gmm] (5) at (0.00, 0.25) {};
    \node [style=gpp] (6) at (0.00, -0.75) {};
    \node [style=red smallvertex] (7) at (0.00, -1.50) {};
    \node [style=gmm] (8) at (0.75, 1.00) {};
    \node [style=rpp] (9) at (1.00, -0.75) {};
    \node [style=gpp] (10) at (1.50, 1.00) {};
    \node [style=boundary vertex] (11) at (2.25, 1.00) {};
    \node [style=boundary vertex] (12) at (2.25, -1.50) {};
  \end{pgfonlayer}
  \begin{pgfonlayer}{edgelayer}
    \draw (6) to (9);
    \draw (4) to (11);
    \draw (1) to (12);
    \draw (4) to (7);
    \draw (0) to (4);
  \end{pgfonlayer}
\end{tikzpicture}}

\newcommand{\CNOTCNOTsCaseAv}{%
\begin{tikzpicture}[quanto,halfsize]
  \begin{pgfonlayer}{nodelayer}
    \node [style=boundary vertex] (0) at (-2.25, 1.00) {};
    \node [style=boundary vertex] (1) at (-2.25, -1.50) {};
    \node [style=gpp] (2) at (-1.50, 1.00) {};
    \node [style=rpp] (3) at (-0.75, 1.00) {};
    \node [style=gmm] (4) at (0.00, 1.00) {};
    \node [style=rpp] (5) at (0.00, 0.25) {};
    \node [style=gpp] (6) at (0.00, -0.75) {};
    \node [style=red smallvertex] (7) at (0.00, -1.50) {};
    \node [style=rpp] (8) at (0.75, 1.00) {};
    \node [style=rpp] (9) at (1.00, -0.75) {};
    \node [style=gpp] (10) at (1.50, 1.00) {};
    \node [style=boundary vertex] (11) at (2.25, 1.00) {};
    \node [style=boundary vertex] (12) at (2.25, -1.50) {};
  \end{pgfonlayer}
  \begin{pgfonlayer}{edgelayer}
    \draw (0) to (4);
    \draw (6) to (9);
    \draw (4) to (7);
    \draw (1) to (12);
    \draw (4) to (11);
  \end{pgfonlayer}
\end{tikzpicture}}

\newcommand{\CNOTCNOTsCaseAvi}{%
\begin{tikzpicture}[quanto,halfsize]
  \begin{pgfonlayer}{nodelayer}
    \node [style=boundary vertex] (0) at (-2.25, 1.00) {};
    \node [style=boundary vertex] (1) at (-2.25, -2.00) {};
    \node [style=gpp] (2) at (-1.50, 1.00) {};
    \node [style=rpp] (3) at (-0.75, 1.00) {};
    \node [style=gmm] (4) at (0.00, 1.00) {};
    \node [style=rpp] (5) at (0.00, 0.25) {};
    \node [style=gpp] (6) at (0.00, -0.50) {};
    \node [style=rpp] (7) at (0.00, -1.25) {};
    \node [style=rmm] (8) at (0.00, -2.00) {};
    \node [style=rpp] (9) at (0.75, 1.00) {};
    \node [style=rpp] (10) at (0.75, -0.50) {};
    \node [style=gpp] (11) at (1.50, 1.00) {};
    \node [style=red smallvertex] (12) at (1.50, -0.50) {};
    \node [style=boundary vertex] (13) at (2.25, 1.00) {};
    \node [style=boundary vertex] (14) at (2.25, -2.00) {};
  \end{pgfonlayer}
  \begin{pgfonlayer}{edgelayer}
    \draw (0) to (4);
    \draw (6) to (10);
    \draw (4) to (8);
    \draw (12) to (10);
    \draw (1) to (14);
    \draw (4) to (13);
  \end{pgfonlayer}
\end{tikzpicture}}

\newcommand{\CNOTCNOTsCaseAvii}{%
\begin{tikzpicture}[quanto,halfsize]
  \begin{pgfonlayer}{nodelayer}
    \node [style=boundary vertex] (0) at (-2.25, 1.00) {};
    \node [style=boundary vertex] (1) at (-2.25, -2.00) {};
    \node [style=gpp] (2) at (-1.50, 1.00) {};
    \node [style=rpp] (3) at (-0.75, 1.00) {};
    \node [style=gmm] (4) at (0.00, 1.00) {};
    \node [style=gmm] (5) at (0.00, 0.25) {};
    \node [style=rpi] (6) at (0.00, -0.50) {};
    \node [style=gmm] (7) at (0.00, -1.25) {};
    \node [style=rmm] (8) at (0.00, -2.00) {};
    \node [style=rpp] (9) at (0.75, 1.00) {};
    \node [style=red smallvertex] (10) at (0.75, -0.50) {};
    \node [style=gpp] (11) at (1.50, 1.00) {};
    \node [style=boundary vertex] (12) at (2.25, 1.00) {};
    \node [style=boundary vertex] (13) at (2.25, -2.00) {};
  \end{pgfonlayer}
  \begin{pgfonlayer}{edgelayer}
    \draw (4) to (8);
    \draw (4) to (12);
    \draw (1) to (13);
    \draw (0) to (4);
    \draw (6) to (10);
  \end{pgfonlayer}
\end{tikzpicture}}

\newcommand{\CNOTCNOTsCaseAviii}{%
\begin{tikzpicture}[quanto,halfsize]
  \begin{pgfonlayer}{nodelayer}
    \node [style=boundary vertex] (0) at (-2.25, 1.00) {};
    \node [style=boundary vertex] (1) at (-2.25, -2.00) {};
    \node [style=gpp] (2) at (-1.50, 1.00) {};
    \node [style=rpp] (3) at (-0.75, 1.00) {};
    \node [style=gmm] (4) at (0.00, 1.00) {};
    \node [style=gmm] (5) at (0.00, 0.25) {};
    \node [style=gpp] (6) at (0.00, -0.50) {};
    \node [style=rpi] (7) at (0.00, -1.25) {};
    \node [style=rmm] (8) at (0.00, -2.00) {};
    \node [style=rpp] (9) at (0.75, 1.00) {};
    \node [style=gpp] (10) at (1.50, 1.00) {};
    \node [style=boundary vertex] (11) at (2.25, 1.00) {};
    \node [style=boundary vertex] (12) at (2.25, -2.00) {};
  \end{pgfonlayer}
  \begin{pgfonlayer}{edgelayer}
    \draw (4) to (8);
    \draw (4) to (11);
    \draw (1) to (12);
    \draw (0) to (4);
  \end{pgfonlayer}
\end{tikzpicture}}

\newcommand{\CNOTCNOTsCaseAix}{%
\begin{tikzpicture}[quanto,halfsize]
  \begin{pgfonlayer}{nodelayer}
    \node [style=boundary vertex] (0) at (-2.25, 1.00) {};
    \node [style=boundary vertex] (1) at (-2.25, -1.00) {};
    \node [style=gpp] (2) at (-1.50, 1.00) {};
    \node [style=rpp] (3) at (-0.75, 1.00) {};
    \node [style=gmm] (4) at (0.00, 1.00) {};
    \node [style=rpp] (5) at (0.00, -1.00) {};
    \node [style=rpp] (6) at (0.75, 1.00) {};
    \node [style=gpp] (7) at (1.50, 1.00) {};
    \node [style=boundary vertex] (8) at (2.25, 1.00) {};
    \node [style=boundary vertex] (9) at (2.25, -1.00) {};
  \end{pgfonlayer}
  \begin{pgfonlayer}{edgelayer}
    \draw (4) to (5);
    \draw (4) to (8);
    \draw (1) to (9);
    \draw (0) to (4);
  \end{pgfonlayer}
\end{tikzpicture}}

\newcommand{\CNOTCNOTsCaseAx}{%
\begin{tikzpicture}[quanto,halfsize]
  \begin{pgfonlayer}{nodelayer}
    \node [style=boundary vertex] (0) at (-3.00, 1.00) {};
    \node [style=boundary vertex] (1) at (-3.00, -1.00) {};
    \node [style=gpp] (2) at (-2.25, 1.00) {};
    \node [style=rpp] (3) at (-1.50, 1.00) {};
    \node [style=rpp] (4) at (-1.50, -1.00) {};
    \node [style=gmm] (5) at (-0.75, 1.00) {};
    \node [style=green smallvertex] (6) at (0.00, 1.00) {};
    \node [style=red smallvertex] (7) at (0.00, -1.00) {};
    \node [style=rpp] (8) at (0.75, 1.00) {};
    \node [style=gpp] (9) at (1.50, 1.00) {};
    \node [style=boundary vertex] (10) at (2.25, 1.00) {};
    \node [style=boundary vertex] (11) at (2.25, -1.00) {};
  \end{pgfonlayer}
  \begin{pgfonlayer}{edgelayer}
    \draw (6) to (7);
    \draw (6) to (10);
    \draw (1) to (11);
    \draw (0) to (6);
  \end{pgfonlayer}
\end{tikzpicture}}

\newcommand{\CNOTCNOTsCaseBi}{%
\begin{tikzpicture}[quanto,halfsize]
  \begin{pgfonlayer}{nodelayer}
    \node [style=boundary vertex] (0) at (-2.00, 1.00) {};
    \node [style=boundary vertex] (1) at (-2.00, -1.00) {};
    \node [style=green smallvertex] (2) at (-1.00, 1.00) {};
    \node [style=red smallvertex] (3) at (-1.00, -1.00) {};
    \node [style=rpp] (4) at (0.00, 1.00) {};
    \node [style=gpp] (5) at (0.00, -1.00) {};
    \node [style=green smallvertex] (6) at (1.00, 1.00) {};
    \node [style=red smallvertex] (7) at (1.00, -1.00) {};
    \node [style=boundary vertex] (8) at (2.00, 1.00) {};
    \node [style=boundary vertex] (9) at (2.00, -1.00) {};
  \end{pgfonlayer}
  \begin{pgfonlayer}{edgelayer}
    \draw (7) to (6);
    \draw (2) to (3);
    \draw (1) to (9);
    \draw (0) to (2);
    \draw (2) to (8);
  \end{pgfonlayer}
\end{tikzpicture}
}

\newcommand{\CNOTCNOTsCaseBii}{%
\begin{tikzpicture}[quanto,halfsize]
  \begin{pgfonlayer}{nodelayer}
    \node [style=boundary vertex] (0) at (-2.00, 1.00) {};
    \node [style=boundary vertex] (1) at (-2.00, -1.00) {};
    \node [style=gpi] (2) at (-1.25, 1.00) {};
    \node [style=gpi] (3) at (-1.25, -1.00) {};
    \node [style=green smallvertex] (4) at (-0.50, 1.00) {};
    \node [style=red smallvertex] (5) at (-0.50, -1.00) {};
    \node [style=gpi] (6) at (0.25, -1.00) {};
    \node [style=rpp] (7) at (0.50, 1.00) {};
    \node [style=gpp] (8) at (1.00, -1.00) {};
    \node [style=green smallvertex] (9) at (1.75, 1.00) {};
    \node [style=red smallvertex] (10) at (1.75, -1.00) {};
    \node [style=boundary vertex] (11) at (2.75, 1.00) {};
    \node [style=boundary vertex] (12) at (2.75, -1.00) {};
  \end{pgfonlayer}
  \begin{pgfonlayer}{edgelayer}
    \draw (10) to (9);
    \draw (4) to (11);
    \draw (1) to (12);
    \draw (4) to (5);
    \draw (0) to (4);
  \end{pgfonlayer}
\end{tikzpicture}}

\newcommand{\CNOTCNOTsCaseBiii}{%
  \begin{tikzpicture}[quanto,halfsize]
    \begin{pgfonlayer}{nodelayer}
      \node [style=boundary vertex] (0) at (-2.50, 1.00) {};
      \node [style=boundary vertex] (1) at (-2.50, -1.00) {};
    \node [style=gpi] (2) at (-1.75, 1.00) {};
    \node [style=gpi] (3) at (-1.75, -1.00) {};
    \node [style=rpi] (4) at (-1.00, 1.00) {};
    \node [style=rpi] (5) at (-1.00, -1.00) {};
    \node [style=green smallvertex] (6) at (-0.25, 1.00) {};
    \node [style=red smallvertex] (7) at (-0.25, -1.00) {};
    \node [style=rpi] (8) at (0.50, 1.00) {};
    \node [style=gpi] (9) at (0.50, -1.00) {};
    \node [style=rpp] (10) at (1.25, 1.00) {};
    \node [style=gpp] (11) at (1.25, -1.00) {};
    \node [style=green smallvertex] (12) at (2.00, 1.00) {};
    \node [style=red smallvertex] (13) at (2.00, -1.00) {};
    \node [style=boundary vertex] (14) at (2.75, 1.00) {};
    \node [style=boundary vertex] (15) at (2.75, -1.00) {};
  \end{pgfonlayer}
  \begin{pgfonlayer}{edgelayer}
    \draw (13) to (12);
    \draw (0) to (6);
    \draw (1) to (15);
    \draw (6) to (14);
    \draw (6) to (7);
  \end{pgfonlayer}
\end{tikzpicture}}

\newcommand{\CNOTCNOTsCaseBiv}{%
\begin{tikzpicture}[quanto,halfsize]
  \begin{pgfonlayer}{nodelayer}
    \node [style=boundary vertex] (0) at (-2.50, 1.00) {};
    \node [style=boundary vertex] (1) at (-2.50, -1.00) {};
    \node [style=rpi] (2) at (-1.75, 1.00) {};
    \node [style=gpi] (3) at (-1.75, -1.00) {};
    \node [style=green smallvertex] (4) at (-1.00, 1.00) {};
    \node [style=red smallvertex] (5) at (-1.00, -1.00) {};
    \node [style=gpi] (6) at (-0.25, 1.00) {};
    \node [style=rpi] (7) at (-0.25, -1.00) {};
    \node [style=rpi] (8) at (0.50, 1.00) {};
    \node [style=gpi] (9) at (0.50, -1.00) {};
    \node [style=rpp] (10) at (1.25, 1.00) {};
    \node [style=gpp] (11) at (1.25, -1.00) {};
    \node [style=green smallvertex] (12) at (2.00, 1.00) {};
    \node [style=red smallvertex] (13) at (2.00, -1.00) {};
    \node [style=boundary vertex] (14) at (2.75, 1.00) {};
    \node [style=boundary vertex] (15) at (2.75, -1.00) {};
  \end{pgfonlayer}
  \begin{pgfonlayer}{edgelayer}
    \draw (0) to (4);
    \draw (1) to (15);
    \draw (4) to (14);
    \draw (4) to (5);
    \draw (13) to (12);
  \end{pgfonlayer}
\end{tikzpicture}}

\newcommand{\CNOTCNOTsCaseBv}{%
\begin{tikzpicture}[quanto,halfsize]
  \begin{pgfonlayer}{nodelayer}
    \node [style=boundary vertex] (0) at (-4.50, 1.00) {};
    \node [style=boundary vertex] (1) at (-4.50, -1.00) {};
    \node [style=rmm] (2) at (-4.00, 1.00) {};
    \node [style=gmm] (3) at (-4.00, -1.00) {};
    \node [style=gpp] (4) at (-3.25, 1.00) {};
    \node [style=rpp] (5) at (-3.25, -1.00) {};
    \node [style=gmm] (6) at (-2.50, 1.00) {};
    \node [style=rmm] (7) at (-2.50, -1.00) {};
    \node [style=rmm] (8) at (-1.75, 1.00) {};
    \node [style=gmm] (9) at (-1.75, -1.00) {};
    \node [style=green smallvertex] (10) at (-1.00, 1.00) {};
    \node [style=red smallvertex] (11) at (-1.00, -1.00) {};
    \node [style=rpp] (12) at (-0.25, 1.00) {};
    \node [style=gpp] (13) at (-0.25, -1.00) {};
    \node [style=gpp] (14) at (0.50, 1.00) {};
    \node [style=rpp] (15) at (0.50, -1.00) {};
    \node [style=gpp] (16) at (1.25, 1.00) {};
    \node [style=rpp] (17) at (1.25, -1.00) {};
    \node [style=green smallvertex] (18) at (2.00, 1.00) {};
    \node [style=red smallvertex] (19) at (2.00, -1.00) {};
    \node [style=boundary vertex] (20) at (2.75, 1.00) {};
    \node [style=boundary vertex] (21) at (2.75, -1.00) {};
  \end{pgfonlayer}
  \begin{pgfonlayer}{edgelayer}
    \draw (0) to (10);
    \draw (1) to (21);
    \draw (10) to (20);
    \draw (10) to (11);
    \draw (19) to (18);
  \end{pgfonlayer}
\end{tikzpicture}}

\newcommand{\CNOTCNOTsCaseBvi}{%
\begin{tikzpicture}[quanto,halfsize]
  \begin{pgfonlayer}{nodelayer}
    \node [style=boundary vertex] (0) at (-4.50, 1.00) {};
    \node [style=boundary vertex] (1) at (-4.50, -1.00) {};
    \node [style=rmm] (2) at (-4.00, 1.00) {};
    \node [style=gmm] (3) at (-4.00, -1.00) {};
    \node [style=gpp] (4) at (-3.25, 1.00) {};
    \node [style=rpp] (5) at (-3.25, -1.00) {};
    \node [style=red smallvertex] (6) at (-1.00, 1.00) {};
    \node [style=green smallvertex] (7) at (-1.00, -1.00) {};
    \node [style=gpp] (8) at (1.25, 1.00) {};
    \node [style=rpp] (9) at (1.25, -1.00) {};
    \node [style=green smallvertex] (10) at (2.00, 1.00) {};
    \node [style=red smallvertex] (11) at (2.00, -1.00) {};
    \node [style=boundary vertex] (12) at (2.75, 1.00) {};
    \node [style=boundary vertex] (13) at (2.75, -1.00) {};
  \end{pgfonlayer}
  \begin{pgfonlayer}{edgelayer}
    \draw (11) to (10);
    \draw (0) to (6);
    \draw (6) to (7);
    \draw (1) to (13);
    \draw (6) to (12);
  \end{pgfonlayer}
\end{tikzpicture}
}
\newcommand{\CNOTCNOTsCaseBvii}{%
\begin{tikzpicture}[quanto,halfsize]
  \begin{pgfonlayer}{nodelayer}
    \node [style=boundary vertex] (0) at (-4.50, 1.00) {};
    \node [style=boundary vertex] (1) at (-4.50, -1.00) {};
    \node [style=rmm] (2) at (-4.00, 1.00) {};
    \node [style=gmm] (3) at (-4.00, -1.00) {};
    \node [style=gpp] (4) at (-3.25, 1.00) {};
    \node [style=rpp] (5) at (-3.25, -1.00) {};
    \node [style=green smallvertex] (6) at (-1.00, 1.00) {};
    \node [style=red smallvertex] (7) at (-1.00, -1.00) {};
    \node [style=green smallvertex] (8) at (1.25, 1.00) {};
    \node [style=red smallvertex] (9) at (1.25, -1.00) {};
    \node [style=gpp] (10) at (2.00, 1.00) {};
    \node [style=rpp] (11) at (2.00, -1.00) {};
    \node [style=boundary vertex] (12) at (2.75, 1.00) {};
    \node [style=boundary vertex] (13) at (2.75, -1.00) {};
  \end{pgfonlayer}
  \begin{pgfonlayer}{edgelayer}
    \draw (4) to (0);
    \draw [in=180, out=0] (7) to (8);
    \draw (9) to (8);
    \draw [in=180, out=0] (4) to (7);
    \draw (8) to (12);
    \draw (6) to (7);
    \draw [in=0, out=180] (6) to (5);
    \draw [in=180, out=0] (6) to (9);
    \draw (9) to (13);
    \draw (1) to (5);
  \end{pgfonlayer}
\end{tikzpicture}
}

\newcommand{\CNOTCNOTsCaseBviii}{%
\begin{tikzpicture}[quanto,halfsize]
  \begin{pgfonlayer}{nodelayer}
    \node [style=boundary vertex] (0) at (-4.50, 1.00) {};
    \node [style=boundary vertex] (1) at (-4.50, -1.00) {};
    \node [style=rmm] (2) at (-4.00, 1.00) {};
    \node [style=gmm] (3) at (-4.00, -1.00) {};
    \node [style=gpp] (4) at (-3.25, 1.00) {};
    \node [style=rpp] (5) at (-3.25, -1.00) {};
    \node [style=red smallvertex] (8) at (-0.50, 1.00) {};
    \node [style=green smallvertex] (9) at (-0.50, -1.00) {};
    \node [style=gpp] (10) at (1.00, 1.00) {};
    \node [style=rpp] (11) at (1.00, -1.00) {};
    \node [style=boundary vertex] (12) at (1.75, 1.00) {};
    \node [style=boundary vertex] (13) at (1.75, -1.00) {};
  \end{pgfonlayer}
  \begin{pgfonlayer}{edgelayer}
    \draw (4) to (2);
    \draw (9) to (8);
    \draw [out=180, in=0] (8) to (5);
    \draw [out=180, in=0] (9) to (4);
    \draw (8) to (12);
    \draw (8) to (10);
    \draw (9) to (13);
    \draw (4) to (0);
    \draw (3) to (5);
  \end{pgfonlayer}
\end{tikzpicture}
}
\newcommand{\CNOTCNOTsCaseBix}{%
\begin{tikzpicture}[quanto,halfsize]
  \begin{pgfonlayer}{nodelayer}
    \node [style=boundary vertex] (0) at (-4.50, 1.00) {};
    \node [style=boundary vertex] (1) at (-4.50, -1.00) {};
    \node [style=rmm] (2) at (-4.00, 1.00) {};
    \node [style=gmm] (3) at (-4.00, -1.00) {};
    \node [style=gpp] (4) at (-3.25, 1.00) {};
    \node [style=rpp] (5) at (-3.25, -1.00) {};
    \node [style=green smallvertex] (6) at (-2.50, 1.00) {};
    \node [style=red smallvertex] (7) at (-2.50, -1.00) {};
    \node [style=gpp] (8) at (-0.50, 1.00) {};
    \node [style=rpp] (9) at (-0.50, -1.00) {};
    \node [style=boundary vertex] (10) at (0.25, 1.00) {};
    \node [style=boundary vertex] (11) at (0.25, -1.00) {};
  \end{pgfonlayer}
  \begin{pgfonlayer}{edgelayer}
    \draw [in=180, out=0] (4) to (6);
    \draw (6) to (7);
    \draw [in=180, out=0] (7) to (5);
    \draw [in=180, out=0] (6) to (9);
    \draw (8) to (10);
    \draw [in=0, out=180] (8) to (7);
    \draw (9) to (11);
    \draw (4) to (0);
    \draw (1) to (5);
  \end{pgfonlayer}
\end{tikzpicture}}

\newcommand{\CNOTswapLHS}{%
\begin{tikzpicture}[quanto,halfsize]
  \begin{pgfonlayer}{nodelayer}
    \node [style=boundary vertex] (0) at (-4.75, 1.00) {};
    \node [style=boundary vertex] (1) at (-4.75, -1.00) {};
    \node [style=green smallvertex] (2) at (-4.00, 1.00) {};
    \node [style=red smallvertex] (3) at (-4.00, -1.00) {};
    \node [style=boxb] (4) at (-2.00, 1.00) {};
    \node [style=boxa] (5) at (-2.00, -1.00) {};
    \node [style=green smallvertex] (6) at (-1.00, 1.00) {};
    \node [style=red smallvertex] (7) at (-1.00, -1.00) {};
    \node [style=boundary vertex] (8) at (0.00, 1.00) {};
    \node [style=boundary vertex] (9) at (0.00, -1.00) {};
  \end{pgfonlayer}
  \begin{pgfonlayer}{edgelayer}
    \draw [in=0, out=180] (4) to (3);
    \draw (6) to (7);
    \draw (5) to (9);
    \draw (2) to (3);
    \draw (4) to (8);
    \draw (0) to (2);
    \draw (1) to (3);
    \draw [in=180, out=0] (2) to (5);
  \end{pgfonlayer}
\end{tikzpicture}}

\newcommand{\CNOTswapRHS}{%
\begin{tikzpicture}[quanto,halfsize]
  \begin{pgfonlayer}{nodelayer}
    \node [style=boundary vertex] (0) at (-4.75, 1.00) {};
    \node [style=boundary vertex] (1) at (-4.75, -1.00) {};
    \node [style=green smallvertex] (2) at (-4.00, 1.00) {};
    \node [style=red smallvertex] (3) at (-4.00, -1.00) {};
    \node [style=boxa] (4) at (-3.00, 1.00) {};
    \node [style=boxb] (5) at (-3.00, -1.00) {};
    \node [style=red smallvertex] (6) at (-2.00, 1.00) {};
    \node [style=green smallvertex] (7) at (-2.00, -1.00) {};
    \node [style=boundary vertex] (8) at (0.25, 1.00) {};
    \node [style=boundary vertex] (9) at (0.25, -1.00) {};
  \end{pgfonlayer}
  \begin{pgfonlayer}{edgelayer}
    \draw (2) to (6);
    \draw (2) to (3);
    \draw (0) to (2);
    \draw (1) to (3);
    \draw (3) to (7);
    \draw (7) to (6);
    \draw [in=180, out=0] (6) to (9);
    \draw [in=180, out=0] (7) to (8);
  \end{pgfonlayer}
\end{tikzpicture}}

\newcommand{\toncSWAP}{%
  \begin{tikzpicture}[quanto,halfsize]
    \node [boundary vertex] (a) at (0.0,-1.0) {};
    \node [boundary vertex] (d) at (-4,-3.0) {};
    \node [boundary vertex] (c) at (-4,-1.0) {};
    \node [boundary vertex] (b) at (0.0,-3.0) {};
    \node [green smallvertex] (f) at (-2.0,-1.0) {};
    \node [red smallvertex] (e) at (-2.0,-3.0) {};
    \draw [] (f) to (e);
    \draw [out=180,in=0] (f) to (d);
    \draw [out=180,in=0] (e) to (c);
    \draw [out=180,in=0] (b) to (f);
    \draw [out=180,in=0] (a) to (e);
  \end{tikzpicture}
}

\newcommand{\cnottonccnotI}{%
  \begin{tikzpicture}[quanto,halfsize]
    \node [boundary vertex] (tl) at (1.0,-1.0) {};
    \node [boundary vertex] (br) at (-5,-3.0) {};
    \node [boundary vertex] (bl) at (-5,-1.0) {};
    \node [boundary vertex] (tr) at (1.0,-3.0) {};
    \node [green smallvertex] (a) at (0.0,-1.0) {};
    \node [red smallvertex] (d) at (-4,-3.0) {};
    \node [green smallvertex] (c) at (-4,-1.0) {};
    \node [red smallvertex] (b) at (0.0,-3.0) {};
    \node [green smallvertex] (f) at (-2.0,-1.0) {};
    \node [red smallvertex] (e) at (-2.0,-3.0) {};
    \draw [] (a) to (b);
    \draw [] (c) to (d);
    \draw [] (f) to (e);
    \draw [out=180,in=0] (f) to (d);
    \draw [out=180,in=0] (e) to (c);
    \draw [out=180,in=0] (b) to (f);
    \draw [out=180,in=0] (a) to (e);
    \draw [] (tr) to (b) ;
    \draw [] (tl) to (a) ;
    \draw [] (br) to (d) ;
    \draw [] (bl) to (c) ;
  \end{tikzpicture}
}

\newcommand{\cnottonccnotII}{%
  \begin{tikzpicture}[quanto,halfsize]
    \node [boundary vertex] (tl) at (1.0,-1.0) {};
    \node [boundary vertex] (br) at (-3,-3.0) {};
    \node [boundary vertex] (bl) at (-3,-1.0) {};
    \node [boundary vertex] (tr) at (1.0,-3.0) {};
    \node [green smallvertex] (a) at (0.0,-1.0) {};
    \node [red smallvertex] (b) at (0.0,-3.0) {};
    \node [red smallvertex] (f) at (-2.0,-1.0) {};
    \node [green smallvertex] (e) at (-2.0,-3.0) {};
    \draw [] (a) to (b);
    \draw [] (f) to (e);
    \draw [out=180,in=0] (b) to (f);
    \draw [out=180,in=0] (a) to (e);
    \draw [] (tr) to (b) ;
    \draw [] (tl) to (a) ;
    \draw [] (br) to (e) ;
    \draw [] (bl) to (f) ;
  \end{tikzpicture}
}

\newcommand{\cnottonccnotIV}{%
  \begin{tikzpicture}[quanto,halfsize]
    \node [boundary vertex] (tl) at (1.0,-1.0) {};
    \node [boundary vertex] (br) at (-3,-3.0) {};
    \node [boundary vertex] (bl) at (-3,-1.0) {};
    \node [boundary vertex] (tr) at (1.0,-3.0) {};
    \node [green smallvertex] (a) at (-0.5,-1.0) {};
    \node [red smallvertex] (b) at (-0.5,-3.0) {};
    \draw [out=-45,in=45] (a) to (b);
    \draw [out=-135,in=135] (a) to (b);
    \draw [] (tr) to (b) ;
    \draw [] (tl) to (a) ;
    \draw [out=0,in=180] (br) to (a) ;
    \draw [out=0,in=180] (bl) to (b) ;
  \end{tikzpicture}
}

\newcommand{\cnottonccnotV}{%
  \begin{tikzpicture}[quanto,halfsize]
    \node [boundary vertex] (tl) at (1.0,-1.0) {};
    \node [boundary vertex] (br) at (-3,-3.0) {};
    \node [boundary vertex] (bl) at (-3,-1.0) {};
    \node [boundary vertex] (tr) at (1.0,-3.0) {};
    \node [green smallvertex] (a) at (-0.5,-1.0) {};
    \node [red smallvertex] (b) at (-0.5,-3.0) {};
    \draw [] (tr) to (b) ;
    \draw [] (tl) to (a) ;
    \draw [out=0,in=180] (br) to (a) ;
    \draw [out=0,in=180] (bl) to (b) ;
  \end{tikzpicture}
}

\newcommand{\swapNF}{%
  \begin{tikzpicture}[quanto,halfsize]
    \node [boundary vertex] (a) at (3.0,-1.0) {};
    \node [boundary vertex] (b) at (3.0,-3.0) {};
    \node [boundary vertex] (c) at (-6,-1.0) {};
    \node [boundary vertex] (d) at (-6,-3.0) {};

    \draw [] (b) to (d);
    \draw [] (a) to (c);

    \node [rpi] (rpiL) at (2.0,-1.0) {};
    \node [gpi] (gpiR) at (2.0,-3.0) {};

    \node [gmm] (gmmL) at (1.0,-1.0) {};
    \node [rmm] (rmmR) at (1.0,-3.0) {};

    \node [green smallvertex] (ft) at (0,-1.0) {};
    \node [red smallvertex] (et) at (0,-3.0) {};
    \draw [] (ft) to (et);

    \node [rpp] (rppL) at (-1.0,-1.0) {};
    \node [gpp] (gppR) at (-1.0,-3.0) {};

    \node [green smallvertex] (f) at (-2.0,-1.0) {};
    \node [red smallvertex] (e) at (-2.0,-3.0) {};
    \draw [] (f) to (e);

    \node [rpp] (rppL) at (-3.0,-1.0) {};
    \node [gpp] (gppR) at (-3.0,-3.0) {};

    \node [gpp] (gppL) at (-4.0,-1.0) {};
    \node [rpp] (rppR) at (-4.0,-3.0) {};

    \node [green smallvertex] (fb) at (-5.0,-1.0) {};
    \node [red smallvertex] (eb) at (-5.0,-3.0) {};
    \draw [] (fb) to (eb);

  \end{tikzpicture}
}

\RequirePackage{array}  

\newenvironment{authors}[1]%
  {\begingroup
   \newcommand\estyle{}%
   \renewcommand\institute[1]%
     {\\\multicolumn{#1}{@{}c@{}}{\scriptsize\begin{tabular}[t]{@{}>{\footnotesize}c@{}}##1\end{tabular}}}%
   \renewcommand\email[1]%
     {\gdef\estyle{\footnotesize\ttfamily}\\##1\gdef\estyle{}}
   \begin{tabular}[t]{@{}*{#1}{>{\estyle}c}@{}}
  }%
  {\end{tabular}%
   \endgroup
  }

\title{Optimising Clifford Circuits with Quantomatic}

\author{  
\begin{authors}{2}
Andrew Fagan$^{1}$ & Ross Duncan$^{1,2}$
\email{andrew.fagan.2014@uni.strath.ac.uk & ross.duncan@strath.ac.uk}
\institute{${}^1$Department of Computer and Information Sciences\\
University of Strathclyde\\
26 Richmond Street, Glasgow, United Kingdom}
\institute{${}^2$Cambridge Quantum Computing Ltd\\
9a Bridge Street, Cambridge, United Kingdom}
\end{authors}
}
\def\titlerunning{Optimising Clifford Circuits}
\def\authorrunning{A. Fagan \& R. Duncan}

\newcommand\projectrepo{\url{https://gitlab.cis.strath.ac.uk/kwb13215/Clifford-Quanto/}\xspace}

\newcommand{\pos}[1]{\ensuremath{\mathrm{pos}(#1)}\xspace}

\begin{document}

\maketitle
\begin{abstract}
  We present a system of equations between Clifford circuits, all
  derivable in the \zxcalculus, and formalised as rewrite rules in the
  Quantomatic proof assistant. By combining these rules with some
  non-trivial simplification procedures defined in the Quantomatic
  tactic language, we demonstrate the use of Quantomatic as a circuit
  optimisation tool. We prove that the system always reduces Clifford
  circuits of one or two qubits to their minimal form, and give
  numerical results demonstrating its performance on larger Clifford
  circuits.
\end{abstract}

\section{Introduction}
\label{sec:introduction}

Remarkable advances in the past two
years have seen quantum computing hardware reach the point where the
deployment of quantum devices for non-trivial tasks is now a near-term
prospect.  However, these machines still suffer from severe
limitations, both in terms of memory size and the coherence time of
their qubits.  It is therefore of paramount importance to extract the
most useful work from the fewest operations: a poorly optimised
quantum program may not be able to finish before it is undone by
noise.

In this paper we study the automated 
optimisation of \emph{Clifford
  circuits}.  Clifford circuits are not universal for quantum
computation -- they are well known to be efficiently simulable by a
classical computer \cite{Scott-Aaronson:2004yf} -- however adding any
non-Clifford gate to the Cliffords yields a set of approximately
universal operations 
hence it is likely that the vast majority of operations in any quantum
program will be Clifford operations, and hence reducing the Clifford
depth and gate count of a circuit will have substantial benefit.

A secondary reason to focus on Cliffords is that the corresponding
subtheory of quantum mechanics is well-understood in terms of the
\zxcalculus \cite{Coecke:2009aa}.  Backens
\cite{1367-2630-16-9-093021} has shown that the \zxcalculus is sound
and complete for \emph{stabilizer quantum theory} -- that is the
fragment of quantum mechanics containing only the Clifford operations
and states which can be produced from them.  While recent extensions
to the \zxcalculus have been proposed which are complete for the
Clifford+T fragment \cite{Jeandel2017A-Complete-Axio} and for the full
qubit theory \cite{NgWang}, both these extensions are significantly
larger and more complex than the stabilizer subtheory, and both
axiomatisations are undergoing rapid development.  Clifford circuits
therefore provide a stable platform to develop techniques which could
later be extended to a universal language.

In this work we present a \emph{circuit optimiser} which can transform
a Clifford circuit into an equivalent circuit with reduced depth and
total number of operations.  The system always reduces Clifford
circuits of one or two qubits to their minimal form, and yields
significant size reductions for larger circuits.

The optimiser has been developed as a \emph{simplification procedure}
in the graphical proof-assistant Quantomatic
\cite{quantomatichome,Kissinger2015Quantomatic:-A-}.  By working
inside the proof-assistant we obtain an important benefit: alongside
the optimised circuit our optimiser produces a formal proof certifying
the correctness of its circuit transformations.

Quantomatic is a flexible graph rewrite engine which is especially
well-adapted to working with the \zxcalculus.  It is designed for
interactive use, and a typical session involves interleaving automated
tactics\footnote{A \emph{tactic} is a short program which automates a
  particular proof strategy.}, backtracking, and manually choosing
rewrites.  However, our circuit optimiser is designed to operate
without user intervention, and we have necessarily made extensive use
of the Quantomatic's powerful tactic combinators and its built-in Python
interpreter.  In doing so, we have constructed the largest and most
sophisticated Quantomatic development to date.

As a graph rewriting engine,  Quantomatic normally inspects only the
local subgraph structure of a term while searching for a rewrite to
apply.  However, quantum circuits have a global causal structure which
must be maintained to preserve the property of ``being a circuit''.  A key 
contribution of this work is the development of techniques to infer
the global structure and use it to guide the rewrite process.


Selinger \cite{selinger2013generators} presented generators and
relations for all the $n$-qubit Clifford groups as a rewrite system
over string diagrams.  This work has some similarities to that: most
notably, we also treat Clifford circuits as a free PROP subject to an
equivalence relation defined on small subgraphs.  However Selinger's
approach is based on transforming circuits to a standard form which is
often larger than the smallest equivalent circuit, whereas we aim for
minimal forms.  Selinger's rewrite system has a much larger number of
rules than the axioms of the \zxcalculus.  Maslov and Roetteler
\cite{Maslov2017Shorter-stabili} demonstrate that all Clifford
unitaries can be implemented in depth no more than 7 stages deep,
using the \CZ as the only two-qubit gate.  The \CZ has formal
disadvantages in \zxcalculus, so in our development we have used the
\CNOT and swap gates instead.  In the examples we have tested, our
optimiser usually attains similar depth over our gate set; however on
some examples it halts while some simplifications are still possible.
Quantomatic has previously been used to produce mechanised correctness
proofs for quantum error correcting codes
\cite{Garvie2017Verifying-the-S,Chancellor2016Coherent-Parity,Duncan:2013lr},
and many of the techniques used in these works are applied here.


We assume that the reader is at least somewhat familiar with both the
\zxcalculus and the Quantomatic system;  an accessible introduction to
both is found in the earlier formalisation effort
\cite{Garvie2017Verifying-the-S}.

\paragraph{Acknowledgements}
\label{sec:acknowledgements}

The authors wish to thank Aleks Kissinger and Hector Miller-Bakewell
for their technical assistance throughout this project.

\paragraph{The Quantomatic project files}
\label{sec:quant-proj-files}
All the proofs which appear in this paper and its appendix are
publicly available as a downloadable Quantomatic project at 
\projectrepo.


\section{The Clifford Group}
\label{sec:code}

Let $\omega = e^{i\pi/4}$.  The $n$th Clifford group is the group of
unitary matrices acting on ${\mathbb C}^{2^n}$, finitely generated by the
matrices
\[
S = 
\begin{pmatrix*}[r]
  1 &0 \\ 0&i
\end{pmatrix*}
\quad
V = \frac{1}{\sqrt 2}
\begin{pmatrix*}[r]
  \omega & \bar{\omega} \\ \bar{\omega} & \omega
\end{pmatrix*}
\quad
\CNOT = 
  \begin{pmatrix*}
    1&0&0&0 \\ 0&1&0&0 \\ 0&0&0&1 \\ 0&0&1&0
  \end{pmatrix*}
\quad
\sigma =
  \begin{pmatrix*}
    1&0&0&0 \\ 0&0&1&0 \\ 0&1&0&0  \\0&0&0&1 
  \end{pmatrix*}
\]
under tensor product and matrix composition.  Observe that $S$ and $V$
are both order 4, the Pauli matrices are given by $Z = S^2$ and $X =
V^2$, and the usual Hadamard matrix is obtained as $H =
\bar{\omega}SVS$.  Throughout this paper we will quotient the group by
global scalar factors, so that $A = zA$ for all non-zero $z \in
{\mathbb C}$.  The resulting quotient yields the $n$th
\emph{reduced} Clifford group, which we denote \C n.  

\C n is finite for all $n$ and the order of the group is given by the
expression\footnote{The only scalars which arise in the quotient are $\omega^k$
  for $k \in {\mathbb N}$ so the usual Clifford group is eight times larger.}:
\[
\sizeof{\C n} \ = \ \prod^n_{j=1} 2(4^j - 1) 4^j \ = \ 2^{2n + 1}(2^{2n} -1)
\sizeof{\C{n-1}}\,.
\]
The order of \C n grows very quickly: $\sizeof{\C1} = 24$,
$\sizeof{\C2} = 11520$, $\sizeof{\C3} = 92827280$, and so on.

The Clifford group \cite{Gottesman:1999to} is equivalently defined as
the \emph{normalizer} of the Pauli group within the unitaries.  Let \P
n be the subgroup of \C n generated by the $Z$ and $X$
operators\footnote{Under our quotient we have $Y = ZX = XZ$ so these
  generators suffice.}.  Then for all $C \in \C n$ and all $P \in \P
n$ we have $ CP = P'C $ for some $P' \in \P n$.
This crucial fact will heavily
exploited in our circuit optimisation procedure.

A more holistic view of the Cliffords is as a class of \emph{circuit
  diagrams} whose gates are the generators of the group, rather than
the group itself.  This can be formalised as a free \dag-PROP
\cite{Lane1965Categorical-Alg,Lack:2004sf} generated by morphisms
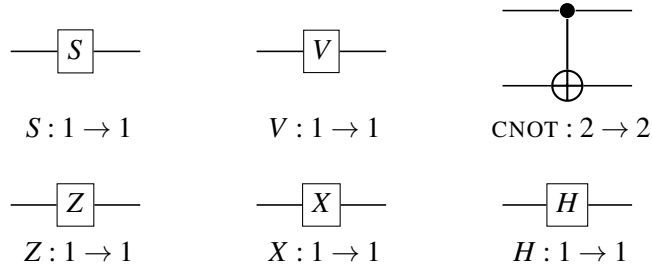
\[
\begin{array}{cccccccccc}
  \circS & \qquad & \circV & \qquad & \circCX  \\ 
  \vstrut{1.2em}
  S : 1 \to 1 && V : 1 \to 1 && \CNOT : 2 \to 2\\
  \\
  \circZ{{}} & \qquad  & \circX{{}} & \qquad & \circH \\ 
  \vstrut{1.2em}
  Z : 1 \to 1 && X : 1 \to 1 && H : 1 \to 1\\
\end{array}
\]
We refer to this PROP as \Cliff.  Note that the swap $\sigma$ is
automatically present due to the PROP structure.  
The morphisms $Z$, $X$, and $H$ are not essential, but it will be
convenient later to include them among the generators. The matrix
valuations above suffice to define a standard interpretation functor
$\denote{\cdot}_C : \Cliff \to \fdhilb$, such that
$\denote{\Cliff(n,n)}_C = \C n$.   
We will use this PROP only to  define the translation from \Cliff to the
\zxcalculus as a functor; we refer the interested reader to Selinger
\cite{selinger2013generators} for more details of this perspective. 

\begin{remark}\label{rem:syntacticPROP}
  Note that \Cliff, as a free PROP, does not include any non-trivial
  equations between Clifford circuits and is not therefore a
  presentation of the Clifford \emph{group} itself:  \Cliff distinguishes
  between different implementations of the same unitary map, which are
  identified in the image of $\denote{\cdot}_C$. In
  particular $\Cliff(n,n)$ is infinite for all $n > 0$.
\end{remark}

\begin{remark}\label{rem:fdhilb}
  Although it will play only a minor role in this paper, we stress
  that \fdhilb denotes the category of finite dimensional Hilbert
  spaces and linear maps under the quotient described above: $A = zA$
  for all non-zero $z \in \mathbb{C}$.  It is common to impose
  this quotient only for $z = e^{i\alpha}$, however in this work the
  question of normalisation will not arise.
\end{remark}

\section{The \zxcalculus}
\label{sec:zxcalculus}

The \zxcalculus \cite{Coecke:2009aa} is a formal graphical
notation for representing quantum states and processes, and an
equational theory for reasoning about them.  It is \emph{universal},
meaning that every linear map has a corresponding \zxcalculus term,
and \emph{sound}, meaning that any equation derivable in the calculus
is true in its standard Hilbert space interpretation.  

The \zxcalculus is formalised using \emph{open graphs}, a
generalisation of the usual notion of graph.  Open graphs have three
kinds of edges: \emph{closed} edges have a vertex at each end, as in a
conventional graph; \emph{half-open} edges are incident upon a single
vertex, the other end being open; and \emph{open} edges are open at
both ends, being incident upon no vertices.  The set of open ends of
the edges form the \emph{boundary} of the open graph.  An open graph
is called \emph{framed} if additionally the boundary is
partitioned into two linearly ordered sets of \emph{inputs} and
\emph{outputs}.

\begin{definition}
  \label{def:zxterms}
  A \emph{term}, or \emph{diagram}, of the \zxcalculus is a finite
  undirected framed open graph whose interior vertices are of
  the following types:
  \begin{itemize}
  \item $Z(\alpha)$ vertices, labelled by an angle $\alpha$ where $0
    \leq \alpha < 2\pi$.  These are depicted as green or light grey
    circles; if $\alpha = 0$  then the label is omitted.
  \item $X(\beta)$ vertices, labelled by an angle $\beta$, where $0
    \leq \beta < 2\pi$.  These are are depicted as red or dark grey
    circles; again, if $\beta = 0$ then the label is omitted.
  \item $H$ vertices; unlike the other types H vertices are constrained
    to have degree exactly 2.  They are depicted as yellow squares.  
  \end{itemize}
Two terms are considered equal if they are isomorphic as framed
labelled graphs.  

\end{definition}
The allowed vertex types are shown in Figure \ref{fig:vertices}.  We
adopt the convention that inputs are on the left, and outputs on the
right.

Compound terms may be formed by joining some number (maybe zero) of
the outputs of one term to the inputs of another.  Given a diagram $D
: n \to m$ we define its \emph{adjoint} $D^\dag : m \to n$ to be the
diagram obtained by reflecting the diagram around the vertical axis
and negating all the angles.  Thus the terms of the \zxcalculus
naturally form a a $\dag$-PROP
\cite{Coecke:2015aa,Duncan2016Interacting-Fro}, just like the
Clifford circuits of the previous section.  The three types of single
vertex shown in Figure~\ref{fig:vertices} can then be seen as the
generators of this PROP, which we call \textbf{ZX}.  

\begin{figure}
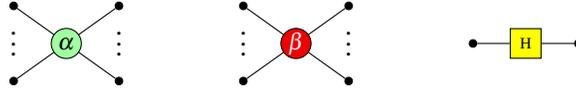

  \centering
  \[
\input{figures/zx-definition/Z.tikz}
\qquad \qquad
\input{figures/zx-definition/X.tikz}
\qquad \qquad
\input{figures/zx-definition/H.tikz} 
\]
  \caption{Allowed interior vertices}
  \label{fig:vertices}
  \figureline
\end{figure}

\begin{remark}
  \label{rem:scalars}
  We consider the \emph{scalar-free} fragment of the \zxcalculus; or
  equivalently, we quotient everything by non-zero scalar factors, just
  as we did in the previous section. This greatly reduces the
  complexity of the diagrams and poses no danger; Backens has
  demonstrated how to modify the calculus to preserve equality while
  respecting scalars \cite{Backens:2015aa}.
\end{remark}

\begin{definition}\label{def:zxsemantics}
  Given a \textsl{zx}-term $D:m\to n$, its \emph{standard
    interpretation} is a linear map $\denote{D} :
  (\mathbb{C}^2)^{\otimes m} \to (\mathbb{C}^2)^{\otimes n}$ defined
  on single vertices as follows:
\[
    \ldenote{\input{figures/zx-definition/Z.tikz}}
     = 
    \left\{
    \begin{array}{ccl}
      \ket{0}^{\otimes m} & \mapsto &   \ket{0}^{\otimes n}\\
      \ket{1}^{\otimes m} & \mapsto &   e^{i \alpha}\ket{1}^{\otimes n}
    \end{array}\right.
\qquad\qquad
    \ldenote{\input{figures/zx-definition/X.tikz}}
     = 
    \left\{
    \begin{array}{ccl}
      \ket{+}^{\otimes m} & \mapsto &   \ket{+}^{\otimes n}\\
      \ket{-}^{\otimes m} & \mapsto &   e^{i \beta}\ket{-}^{\otimes n}
    \end{array}\right.
\]
\vstrut{0.5\baselineskip}
\[
    \ldenote{\input{figures/zx-definition/H.tikz}}
    = 
    \frac{1}{\sqrt{2}}
    \left(
      \begin{array}{cc}
        1&1\\1&-1
      \end{array}
    \right).
\]
Since the above diagrams are the generators of \textbf{ZX}, the above
definition extends uniquely to a strict $\dag$-symmetric monoidal
functor $\denote{\cdot} : \mathbf{ZX} \to \fdhilb$.
\end{definition}


The \zxcalculus has a rich equational theory based on the theory of
Frobenius-Hopf algebras
\cite{Coecke:2009aa,Duncan2016Interacting-Fro}.  Various
axiomatisations have been proposed
(\cite{Duncan:2009ph,EPTCS171.5,Backens2016A-Simplified-St,Perdrix:2015aa,Jeandel2017A-Complete-Axio,NgWang})
with various advantages and drawbacks.  Here we adopt the scheme of
Backens \cite{1367-2630-16-9-093021} which is clean, concise, and
adequate for the treatment of the Clifford group.  These are shown in 
Figure \ref{fig:zxrules}.  Note that, due to the $H$-commute rule, the
colour swapped versions of all the rules are admissible, and we shall
use these colour swapped versions without further comment.

\begin{definition}\label{def:equiv}
  Let $\leftrightarrow$ be the one-step rewrite relation on the terms
  of \ZX generated by the pairs of terms shown in
  Figure~\ref{fig:zxrules};  let $\eq$ be the least equivalence
  relation containing $\leftrightarrow$.  We say that \ZX terms
  $a$ and $b$ are \emph{equivalent} when $a \eq b$.
\end{definition}

We reserve the notation $a = b$ for the case where $a$ and $b$ are
equal as graphs, in the sense of Definition~\ref{def:zxterms}.

\begin{figure}[t]
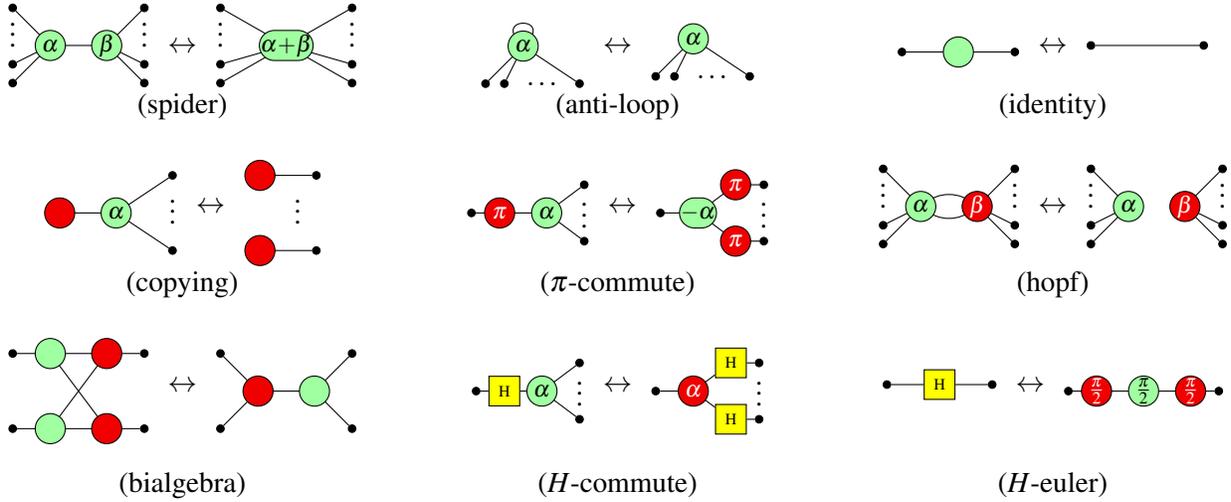


  \makebox[\textwidth][c]{
  \begin{minipage}{1.1\linewidth}
    \[
    \begin{array}{ccccc}
      \input{figures/zx-definition/axioms/spider_1.tikz}
      \;\;\leftrightarrow\;\;
      \input{figures/zx-definition/axioms/spider_2.tikz}
      & 
      \qquad\quad
      &
      \input{figures/zx-definition/axioms/anti-loop_1.tikz}
      \;\;\leftrightarrow\;\;
      \input{figures/zx-definition/axioms/anti-loop_2.tikz}
      &
      \qquad \quad
      &
      \input{figures/zx-definition/axioms/identity_1.tikz}
      \;\;\leftrightarrow\;\;
      \input{figures/zx-definition/axioms/identity_2.tikz}
      \\
      \text{(spider)} && \text{(anti-loop)} && \text{(identity)} 
      \\ \\
      \input{figures/zx-definition/axioms/copying_1.tikz}
      \;\;\leftrightarrow\;\;
      \input{figures/zx-definition/axioms/copying_2.tikz}
      &
      \qquad 
      &
      \input{figures/zx-definition/axioms/pi-commute_1.tikz}
      \;\;\leftrightarrow\;\;
      \input{figures/zx-definition/axioms/pi-commute_2.tikz}
      &
      \qquad 
      &
      \input{figures/zx-definition/axioms/hopf_1.tikz}
      \;\;\leftrightarrow\;\;
      \input{figures/zx-definition/axioms/hopf_2.tikz}
      \\ 
      \text{(copying)}  &&  \text{($\pi$-commute)} &&\text{(hopf)} 
      \\ \\
      \input{figures/zx-definition/axioms/bialgebra_1.tikz}
      \;\;\leftrightarrow\;\;
      \input{figures/zx-definition/axioms/bialgebra_2.tikz}
      & 
      \qquad
      & 
      \input{figures/zx-definition/axioms/h-commute_1.tikz}
      \;\;\leftrightarrow\;\;
      \input{figures/zx-definition/axioms/h-commute_2.tikz}
      & 
      \qquad 
      &
      \input{figures/zx-definition/H.tikz}
      \;\;\leftrightarrow\;\;
      \input{figures/zx-definition/axioms/h-euler-2.tikz}
      \\ \\[-1ex]
      \text{(bialgebra)}  &&\text{($H$-commute)}  &&\text{($H$-euler)} 
    \end{array}
    \]
  \end{minipage}}
  \caption{The rules of the \zxcalculus.  Note all arithmetic is
    modulo $2\pi$.}
  \label{fig:zxrules}
  \figureline
\end{figure}

Thanks to the $H$-euler rule, the $H$ vertices are redundant.  With
this in mind, the following colour changing rule will be useful.

\begin{lemma}\label{lem:colourchange}
  The following rule is admissible:
  \[
  \inline{\input{gen-colour-change-lhs.tikz}} \eq
  \inlinetikzfig{gen-colour-change-rhs}  
  \]
where $n$ is the number of red vertices in the left hand side diagram.
\end{lemma}
\noindent
The proof is given in Appendix \ref{sec:proofs-omitted-from}.


\section{Representing Clifford Circuits in the \zxcalculus}
\label{sec:form-cliff-circ}

For the rest of the paper we will use only the
\emph{stabilizer} sub-language of the full \zxcalculus.

\begin{definition}\label{def:stabZX}
  Let \ZXs denote the sub-PROP of \ZX obtained by restricting to those
  terms whose $Z$ and $X$ vertices are labelled only by angles
  $\alpha \in \{0, \pi/2, \pi, 3\pi/2\}$ .
\end{definition}

Since only four vertex labels can occur in the terms of \ZXs, we will
adopt a more compact notation as shown in below.
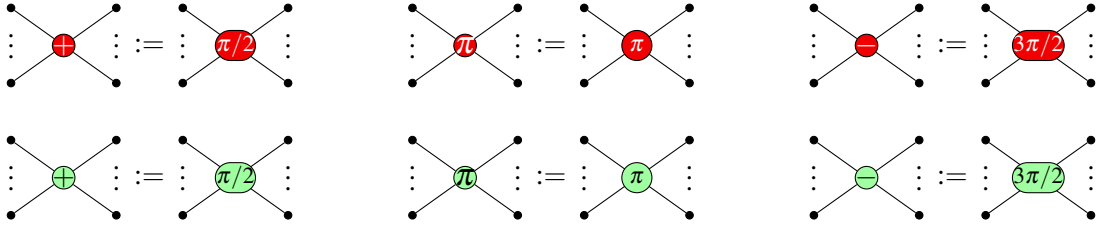
\[
\input{figures/zx-definition/rpp.tikz}
\ := \  
\input{figures/zx-definition/X_pi2.tikz}
\qquad \qquad
\input{figures/zx-definition/rpi.tikz}
\ := \  
\input{figures/zx-definition/X_pi.tikz}
\qquad \qquad
\input{figures/zx-definition/rmm.tikz}
\ := \  
\input{figures/zx-definition/X_3pi2.tikz}
\]
\[
\input{figures/zx-definition/gpp.tikz}
\ := \  
\input{figures/zx-definition/Z_pi2.tikz}
\qquad \qquad
\input{figures/zx-definition/gpi.tikz}
\ := \  
\input{figures/zx-definition/Z_pi.tikz}
\qquad \qquad
\input{figures/zx-definition/gmm.tikz}
\ := \  
\input{figures/zx-definition/Z_3pi2.tikz}
\]
Since the one-qubit instances of these will occur very frequently in
the rest of the paper, we introduce the names:
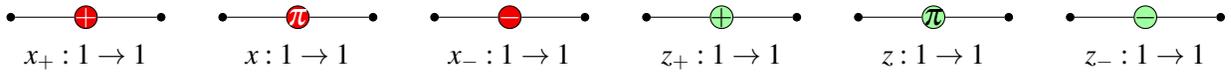
\[
\begin{array}{cccccccccccc}
  \rpp & \qquad & \rpi & \qquad & \rmm  & \qquad &
  \gpp & \qquad  & \gpi & \qquad & \gmm \\

  \vstrut{1.2em}
  x_+ : 1 \to 1 && x : 1 \to 1 && x_- : 1 \to 1 &&
  z_+ : 1 \to 1 && z : 1 \to 1 && z_- : 1 \to 1\\
\end{array}
\]
We refer to $x$ and $z$ as the \emph{Paulis}.
 
We define the \emph{translation} functor $T:\Cliff \to \ZXs$ on the
generators as shown below.
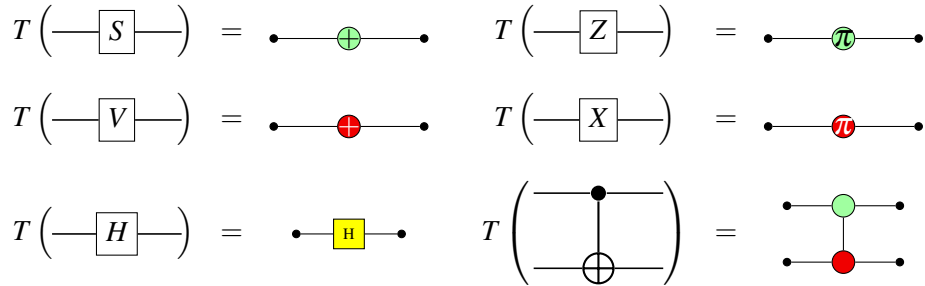
\[
\begin{array}{lcccccccccc}
  T\left(\mkern-18mu \circS \mkern-25mu\right) & = & \gpp  & \qquad &
  T\left(\mkern-18mu  \circZ{{}} \mkern-25mu\right) & = & \gpi & \qquad \\\\
  T\left(\mkern-18mu  \circV \mkern-25mu\right) & = &\rpp & \qquad &
  T\left(\mkern-18mu  \circX{{}} \mkern-25mu\right) & = & \rpi & \qquad \\\\
  T\left(\mkern-18mu  \circH \mkern-25mu\right) & = & \input{figures/zx-definition/H.tikz}  & \qquad &
  T\left(\mkern-15mu  \circCX \mkern-22mu\right) & = & \input{figures/zx-definition/CNOT.tikz} & \qquad \\\\
\end{array}
\]
This extends to the whole PROP in the usual way.  Further we have the following:
\begin{proposition}
  \label{prop:sanity}
  For all $c : n\to m$ in \Cliff we have $\denote{c}_C = \denote{T(c)}$.
\end{proposition}

\noindent
Subject to the proviso that the variables $\alpha$ and $\beta$
occurring in the rules of Figure \ref{fig:zxrules} must take values in
the set $\{0, \pi/2, \pi, 3\pi/2\}$, it is immediate that if $a
\eq b$ then $a$ is in \ZXs if and only if $b$ is. Further, the
equational theory of \ZXs is sound and complete for its standard
interpretation.

\begin{theorem}[Backens \cite{1367-2630-16-9-093021}]
  \label{thm:sound-and-complete}
  Let $a$ and $b$ be terms of \ZXs. Then $a \eq b$ if and only if
  $\denote{a} = \denote{b}$.
\end{theorem}

\noindent
This theorem guarantees that if a given Clifford circuit has a smaller
equivalent circuit then there is a proof of their equivalence in the
\zxcalculus.  The challenge is to find it.  We start by considering
minimal forms for \C 1 and \C 2.

\subsection{1-Qubit Cliffords}
\label{sec:1-qubit-cliffords}

Let \CC1 denote the 24 diagrams shown Figure~\ref{fig:cc1}.  Note that
all of them are in the image of the translation functor $T$, so they
correspond to Clifford circuits.  Further, each of the diagrams has a
distinct interpretation under $\denote\cdot$, so they cover \C 1 and,
by soundness, they are not equivalent.

\begin{figure}[t]
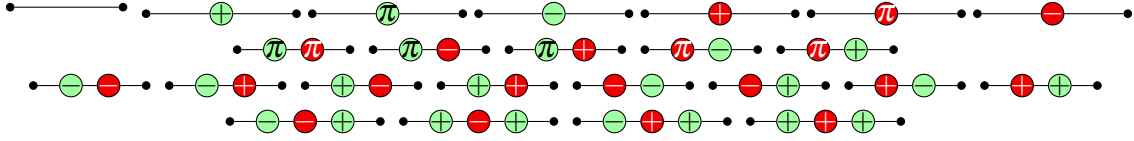

  \centering

\input{zx-definition/axioms/identity_2.tikz}
\gpp
\gpi
\gmm
\rpp
\rpi
\rmm

\input{cliffords/CC1_2_gpirpi.tikz}
\input{cliffords/CC1_2_gpirmm.tikz}
\input{cliffords/CC1_2_gpirpp.tikz}
\input{cliffords/CC1_2_rpigmm.tikz}
\input{cliffords/CC1_2_rpigpp.tikz}

\input{cliffords/CC1_2_gmmrmm.tikz}
\input{cliffords/CC1_2_gmmrpp.tikz}
\input{cliffords/CC1_2_gpprmm.tikz}
\input{cliffords/CC1_2_gpprpp.tikz}
\input{cliffords/CC1_2_rmmgmm.tikz}
\input{cliffords/CC1_2_rmmgpp.tikz}
\input{cliffords/CC1_2_rppgmm.tikz}
\input{cliffords/CC1_2_rppgpp.tikz}

\input{cliffords/CC1_3_pmm.tikz}
\input{cliffords/CC1_3_pmp.tikz}
\input{cliffords/CC1_3_ppm.tikz}
\input{cliffords/CC1_3_ppp.tikz}

  \caption{\CC 1 : standard  minimal forms for single-qubit Cliffords}
  \label{fig:cc1}
  \figureline
\end{figure}


\begin{proposition}\label{prop:c1-nf}
  Given $t \in \Cliff(1,1)$ then $T(t) \eq c$ for some $c \in \CC 1$.
\begin{proof}
  From its type, we have $\denote{t}_C \in \C1$ and since
  $\sizeof{\C1} = 24$, necessarily $\denote{t}_C = \denote{c}$ for
  some $c \in \CC1$, and thus by Proposition~\ref{prop:sanity}
  $\denote{T(t)} = \denote{c}$.  From this point the result follows by
  Theorem~\ref{thm:sound-and-complete}.
\end{proof}
\end{proposition}

\noindent
Since we will require an effective procedure later, we offer an
alternative, constructive, proof which produces the required $c$.  The
ideas of this proof form an important part of our optimisation
procedure.

Recall that a \emph{line graph} is a framed open graph with one input
and one output, which is connected and whose vertices all have
degree two.  In other words, the graph has the shape of a line. We can
treat a line graph as a sequence of vertices starting from the input;
let $\pos v$ denote the sequence position of vertex $v$, \ie 
if $v$ is the first vertex then $\pos v = 0$.

\begin{lemma}\label{lem:lineNF}
  Let $t\in \ZXs(1,1)$ be a line graph; then $t \eq s$ for some
  $s\in\ZXs(1,1)$ which has the following properties:
  \begin{itemize}
  \item its vertices form an alternating sequence of $Z(\alpha)$ and
    $X(\beta)$ vertices, where none of the labels are zero;
  \item it contains at most two Pauli vertices, and if present these
    are first vertices after the input.
  \end{itemize}
  We say that $s$ is in \emph{Pauli-standard} form.
  \begin{proof}
    Observe that by applying the $H$-euler rule all the $H$
    vertices can be removed.  Then, by repeated application of the 
    spider and identity rules, we obtain a graph of alternating
    coloured vertices, none of which are labelled by zero.  Now we
    proceed by induction; define 
    \[
    m(t) = \sum_{v \text{ is a Pauli}} \pos v
    \]
    Observe that if $m(t) < 2$ then $t$ has the required form.
    Otherwise there is at least one Pauli vertex $v$ with $\pos v \geq
    2$. Suppose that $v$ is a Pauli Z vertex; then we may rewrite $t
    \eq t'$ by the following rewrites:
    \[
    \inlinetikzfig{cliffords/pauli-standard-i}
    \rw
    \inlinetikzfig{cliffords/pauli-standard-ii}
    \rw 
    \inlinetikzfig{cliffords/pauli-standard-iii}
    \]
    and applying the identity rule if possible.  Note that $Z(\alpha +
    \pi)$ is never a Pauli, so the number of Paulis is not increased,
    and the number vertices overall is reduced.  If $X(\gamma-\beta)$
    is a Pauli then its position is $\pos v - 1$.  Hence $m(t') <
    m(t)$, and result is achieved by induction.  If $v$ is a Pauli X
    vertex,  then the same argument goes through with the colours
    reversed.
  \end{proof}
\end{lemma}

\begin{lemma}\label{lem:hadamard-equiv}
  We have the following equations:
  \begin{gather*}
    \HcliffI \eq \HcliffII \eq \HcliffV \eq \HcliffVI \\
    \HcliffIb \eq \HcliffIIb \eq \HcliffIVb \eq \HcliffIIIb \\
    \HcliffIc \eq \HcliffIIc \eq \HcliffIVc \eq \HcliffIIIc \\
    \HcliffId \eq \HcliffIId \eq \HcliffIVd \eq \HcliffIIId 
  \end{gather*}
  \begin{proof}
    For the first sequence of equations, consider the following
    rewrites 
    \begin{multline*}
      \HcliffI \rw \HcliffII \rw \HcliffIII \\ \rw \HcliffIV \rw \HcliffV
      \rw \HcliffVI
    \end{multline*}
    where the first and last steps use Lemma
    \ref{lem:colourchange}.  The other three cases are established
    using very similar reasoning.
  \end{proof}
\end{lemma}

\begin{proof}[Alternative proof of Prop \ref{prop:c1-nf}]
  First observe that for any $t \in \Cliff(1,1)$ its image $T(t)$ has
  the form of a line graph.  Now,  by Lemma \ref{lem:lineNF}, $T(t)
  \eq t_1$ where  $t_1$ is in Pauli-standard form.  Now we 
  proceed by induction on $n$, the number of vertices of $t_1$.  

  If $n=0$ or $1$ then $t_1$ is already in \CC1.  If $n=2$ and $t_1$
  has no Pauli nodes, then $t_1$ is already in \CC1; otherwise $t_1
  \eq c$ for some $c$ in \CC1 by the $\pi$-commute rule.

  If $n=3$ and $t_1$ contains at least one Pauli vertex, then by
  applying the $\pi$ commute and spider rules, $t_1$ can rewrite to a
  term with only two vertices or fewer.  For example 
  \begin{equation}\label{eq:pauli-absord}
    \inlinetikzfig{cliffords/reduce-paulis-i}
    \rw
    \inlinetikzfig{cliffords/reduce-paulis-ii}
  \end{equation}
  Otherwise, none of the vertices are Paulis; there
  are 16 line graphs of this form; by Lemma \ref{lem:hadamard-equiv}
  these are partitioned in four equivalence classes. Each of the
  four equivalence classes contains an element of \CC1, hence $t_1 \eq
  c \in \CC1$.

  Finally, suppose $n \geq 4$.  If $t_1$ has any Pauli vertices, then
  $t_1 \eq t_2$ by (\ref{eq:pauli-absord}) where $t_2$ has strictly
  fewer Pauli vertices, and strictly fewer vertices overall.
  Otherwise, $t_1$ has no Pauli vertices; therefore any sequence of
  vertices of length three can be rewritten, using Lemma
  \ref{lem:hadamard-equiv}, to a sequence with the colours exchanged.
  Then by applying the spider and identity rules, and Lemma
  \ref{lem:lineNF} $t_1 \rw t_2$ where $t_2$ is in Pauli standard form
  and has strictly fewer vertices than $t_1$.  Hence we obtain the
  result by induction.
\end{proof}

\noindent
By enumerating all diagrams with fewer than four vertices, it's easy
to check that no $c\in\CC 1$ is equivalent to a smaller diagram, so
these diagrams are the minimal forms for \C 1.

\subsection{2-Qubit Cliffords}
\label{sec:2-qubit-cliffords}

\begin{figure}[t]

  \begin{center}
    \input{figures/cliffords/CC2-Id.tikz} \qquad
    \input{figures/cliffords/CC2-SWAP.tikz} \qquad
    \input{figures/cliffords/CC2-CNOT.tikz} \qquad
    \input{figures/cliffords/CC2-TONC.tikz}
  \end{center}
  where $C_1,C_2 \in \CC 1$, $A \in \mathfrak{A}$ and $B \in \mathfrak{B}$.
  \[
  \mathfrak{A} = \{
  \input{zx-definition/axioms/identity_2.tikz}\,,\, 
  \rpp\,,\,
  \inlinetikzfig{cliffords/CC1_2_rppgpp}\,
  \}
  \qquad
  \mathfrak{B} = \{
  \input{zx-definition/axioms/identity_2.tikz}\,,\,
  \gpp\,,\,
  \inlinetikzfig{cliffords/CC1_2_gpprpp} \,
  \}
  \]
  \caption{\CC 2 : standard  minimal forms for two-qubit Cliffords}
  \label{fig:cc2}
  \figureline
\end{figure}

Let \CC2 be the set of diagrams defined in Figure~\ref{fig:cc2}.
Notice that \CC2 has $24^2(1+1+9+9) = 11520$ elements.  As in the case
of \CC1, all the elements of \CC2 are evidently circuits, and it can
be mechanically checked that they are all distinct under the standard
interpretation.  Hence we adopt \CC2 as the normal forms for \C2.
Again, by the completeness of the axioms
(Thm~\ref{thm:sound-and-complete}) we can immediately conclude that
every 2-qubit Clifford circuit is equivalent to some $k \in \CC2$.
Zooming in a little closer we have the following:

\begin{proposition}\label{prop:CC2}
    Let $k\in \CC2$ and let $g$ be a generator of $\C2$ as a \zxcalculus
  term; then $g\circ k \eq k'$ for some $k' \in \CC2$.
\end{proposition}
\begin{proof}
  Notice that the 1-qubit Cliffords $C_1$ and $C_2$ which occur in $k$
  play no part in the computation, hence we ignore them.  
  
  There are effectively four choices for $g$ : $(u\otimes \id{})$,
  $(\id{}\otimes u)$, \CNOT, or  $\sigma$, where $u$ ranges over the
  generators of \CC1.  If $g = \sigma$, the result is trivial.

Further, for
  any $u \in \CC1$, we have $(u \otimes \id{})\circ \CNOT = (a \otimes
  \id{})\circ\CNOT\circ(u_1 \otimes u_2)$ for some $u_1,u_2 \in \CC1$
  and $a\in\mathfrak{A}$, either by commuting vertices of the same
  colour, or taking out factors of Paulis.  Similarly $(\id{} \otimes
  u)\circ \CNOT = (\id{} \otimes b)\circ\CNOT\circ(u_1 \otimes u_2)$
  for some $b\in\mathfrak{B}$.  Hence the only non-trival cases to
  consider are:
\begin{enumerate}
\item $g = \CNOT$ and $k = (a\otimes b)\circ \CNOT$.  There are nine
  possibilities for this case, of which we consider two here. If $a =
  x_+$ and $b = \id{}$, then by using the bialgebra rule we reduce as
  shown in (i) below.  If  $a = x_+$ and $b = z_+$, then a similar
  argument yields (ii).
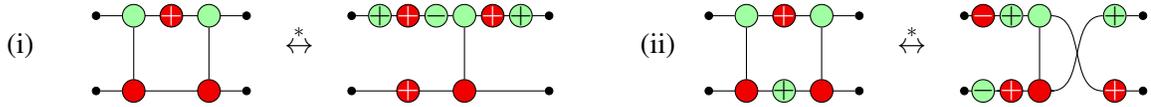
\[
\text{(i)} \qquad
\CNOTCNOTsCaseAi \quad\eq\quad  \CNOTCNOTsCaseAx
 \qquad\quad \text{(ii)} \quad 
\CNOTCNOTsCaseBi \quad\eq\quad  \CNOTCNOTsCaseBix
\]
They remaining cases can all be derived from these two, except the
  case where $a = b = \id{}$, which is trivial.
\item $g = \CNOT$ and $k = (b\otimes a) \circ \sigma \circ \CNOT$.  In
  this case we have (iii) below.  However, by exploiting the $H$ rules
  we have the equivalence (iv) so this case can be reduced to the
  previous one.
  \[
  \text{(iii)}\qquad\CNOTswapLHS \eq \CNOTswapRHS
  \qquad\quad \text{(iv)} \quad 
  \input{zx-definition/TONC.tikz}  \quad \eq \quad \toncNF
  \]
\item $g = \CNOT$ and $k = \sigma$.  Again we can apply (iv).
\end{enumerate}
\end{proof}
As in the single qubit case these reductions will provide the basis
for the rewrite rules of our optimiser.

\section{Quantomatic}
\label{sec:quantomatic}

We briefly sketch the Quantomatic system; for a full description see
\cite{Kissinger2015Quantomatic:-A-}; to obtain it, see
\cite{quantomatichome}.

Quantomatic is an interactive theorem prover 
which can prove equations between terms of 
the \zxcalculus.  
The user draws the desired term in the graphical editor, and builds
the proof by applying \emph{rewrite rules} to the current graph.  A
rewrite rule is a \emph{directed} equation, \ie a pair of graphs which have
the same boundary.  Rules are applied in two-step process.  First the
system searches for subgraphs of the term which are isomorphic to the
LHS of the rule; each such subgraph is called a \emph{match}.
Quantomatic displays all the possible matches of the chosen rule in
the current term, and allows the user to select where the rule is to
be applied.  The matched subgraph is then replaced by the RHS of the
rule to produce a new term.  A proof therefore consists of a sequence
of terms linked by the application of a particular rule at a
particular location in the term.

Quantomatic allows the user to define automated proof tactics, known
as simplification procedures or \emph{simprocs}.  The name is slightly
misleading since there is no need for a simproc to actually
\emph{simplify} the graph, in the sense of reducing the number of
vertices or edges; however failure to do so will typically
result in non-termination.  Simprocs are written in the Python
programming language, interpreted by the built-in Jython interpreter
\cite{jython}, and augmented with various tactic
combinators.  The most useful of these are:
\begin{itemize}
\item REWRITE($r$) : given a rule $r$, apply the rewrite to the first
  match found.
\item REWRITE\_METRIC($r$, $m$) : given a rule $r$ and a \emph{metric}
  function $m : \ZX \to \mathbb{N}$ apply the rewrite to the first
  match found which \emph{reduces} the metric.
\item REWRITE\_TARGETED($r$, $v$, $t$) ; given a rule $r$, a vertex
  $v$ in the LHS of the rule, and a \emph{targeting} function $t : \ZX \to
  \mathrm{Vert}$ apply the rewrite $r$ to the first match where
  vertex $v$ of the rule is matched to the vertex $t(G)$ in the term.
\end{itemize}
All of these combinators can also accept a list of rules, in which
case the
rewrites are attempted in the order of the list.  Simprocs can be
combined in sequence, and also using the combinator REDUCE($s$) which
repeats the simproc $s$ until no new rewrites can be performed.

An important point is that when multiple matches are obtained, the
system will select one without user intervention.  Which rewrite is
selected depends on the internal representation of the term, and is
effectively non-deterministic.  

The axioms of the \zxcalculus, considered as a rewrite system, are
neither terminating nor confluent: for example, if $G \to G'$ via the
$\pi$-commute rule then the rule can be applied again immediately to
rewrite $G' \to G$.  This and similar cases can easily lead to a
non-terminating rewrite strategy.  Further although many of the axioms
can be oriented in a direction which reduces the graph complexity, the
resulting rewrite system is no longer complete with respect to the
standard interpretation: in some cases it is necessary to increase the
graph size to derive a desired equation.

\section{Optimising Clifford Circuits with Quantomatic}
\label{sec:optimising}

In this section we describe our optimisation procedure and summarise
the results.  The reasons for our choices are discussed in the
following section.

By \emph{circuit} we mean a \zxcalculus term which is the image of the
translation functor $T : \Cliff \to \ZX$ defined in Section
\ref{sec:form-cliff-circ}.  Note that since \Cliff is a \dag-PROP and
$\denote{\cdot}$ is a strict monoidal functor, the circuits form a
sub-PROP of \ZX: they are closed under composition and tensor product;
however they are not closed with respect to the equivalence relation $\eq$.

Most \zxcalculus terms are \emph{not} circuits: some of these don't
correspond to Clifford unitaries, and so do not concern us, however
there are many non-circuits which are equivalent to circuits via
rewriting.  We don't attempt to find circuit forms for such terms.
Rather, knowing that the input is a circuit, we use rewrites which
replace circuits with circuits, and thus stay at all times within the
class of circuits.

The axioms of the \zxcalculus calculus are mostly equations between
graphs which are evidently not circuits.  Therefore our first step is
to derive new equations between circuits which will serve as rewrites
for use in the main simproc.  These rules are listed in Appendix
\ref{sec:main-rules}.  

While the derived rules are ``obviously'' equations between circuits,
the circuit structure of a larger graph may not be apparent by looking
at a local subgraph.  It's necessary to look at the entire graph in order
to guide the rewrite engine.





\begin{definition}\label{def:simple-diagram}
  A \zxcalculus diagram is called \emph{simple} if it satisfies
  the following criteria:
  \begin{itemize}
  \item It is a simple graph.
  \item No vertex is adjacent to a vertex of the same colour.
  \item Every degree 2 vertex has a non-zero label.
  \end{itemize}  
\end{definition}
\noindent 

\begin{lemma}\label{lem:simple-diagrams}
  Every \zxcalculus diagram is equal to a simple one.
  \begin{proof}
    Starting with an arbitrary diagram $d$, first replace all $H$
    boxes using the H-euler rule.  Then apply the spider rule to fuse
    spiders, the hopf rule to remove parallel edges, the anti-loop
    rule to remove any self loops, and the identity to remove trivial
    spiders everywhere.
  \end{proof}
\end{lemma}

\begin{definition}{\cite{Danos2006Determinism-in-}}\label{def:causal-flow}
Given a diagram $d$, with vertices $V$, inputs $I$, and outputs $O$, a
\emph{causal flow} $(f,\prec)$ on $d$  consists of a
function $f: I+V \to V+O$ and a partial order $\prec$ on the set
$I+V+O$ satisfying the following properties:
\begin{enumerate}[label=(F\arabic*)]
\item $f(v) \sim v$ \label{flowi}
\item $v \prec f(v)$ \label{flowii}
\item If $u \sim f(v)$ then $v \prec u$ \label{flowiii}
\end{enumerate}
where $v \sim u$ means that $u$ and $v$ are adjacent in the graph.
\end{definition}

In the one-way model \cite{Raussendorf-2001}, the existence of a
causal flow on a graph implies that measurement-based quantum
computation on the corresponding graph state is deterministic
\cite{Danos2006Determinism-in-}.  It plays a similar role for
\zxdiagrams.

\begin{definition}\label{def:circuit-like}
  A simple \zxdiagram is called \emph{circuit-like} if has
  the same number of inputs and outputs and its underlying graph
  admits a \emph{causal flow}.
\end{definition}

\begin{proposition} \label{prop:circuit-like-is-circuit}
  Every circuit-like diagram corresponds to a quantum circuit.
  \begin{proof}
    Using the given flow $(f, \prec)$ we will translate a circuit-like
    diagram $d$ into a circuit $c$ such that $c$ is obvious in the
    image to of $T$.

    We must first determine which edges in the diagram correspond to
    qubits and which correspond to interactions between qubits.  Note
    that for each input $i$, the flow function $f$ defines a path
    through the diagram $i \to f(i) \to f(f(i)) \to \cdots$.  By
    \ref{flowii} and the fact that $\prec$ is a partial order, such a
    path must terminate at an output; and by \ref{flowiii} each such
    input-to-output path is disjoint.  Further, since the numbers of inputs
    and outputs are the same then every vertex of the diagram must
    occur on one such path.

    Now we translate the vertices into gates tracing along the qubit
    paths.  If the vertex has degree 2 it immediately yields a
    $Z_\alpha$ or an $X_\beta$ gate.  If vertex $v$ has higher degree,
    let its neighbours, excluding its immediate predecessor and
    successor, be $u_1, \ldots, u_n$.  We decompose $v$ into $v_0,
    v_1, \ldots v_n$ as shown below:
    \ctikzfig{spider-decompose-Z}
    Note that the new vertices must be compatible with the original
    ordering, in the sense that if $u_i \sim v \sim u_j$ where $u_i
    \prec u_j$ then we require that $v_i \prec v_j$.  By the flow
    conditions, none of the $u_i$ belong to the same qubit as $v$, and
    by the simplicity of the diagram all the $u_i$ are the opposite
    colour to $v$.  Note that the updated diagram also admits a causal
    flow, obtained from the original one in the obvious manner.  Hence
    all the higher degree vertices may be decomposed in the same way,
    to leave the diagram with no vertex of higher degree than three,
    and non-zero phases only in the degree two vertices.  Then all the
    links between qubits maybe replaced by \CNOT gates; since $\prec$
    defines a partial order on the vertices, we are guaranteed this
    is well-defined.
  \end{proof}
\end{proposition}

Intuitively the flow path cover specifies which edges of the
\zxdiagram are the ``qubits'' of the circuit; those edges not in the
path cover define a \CNOT acting on two qubits.  Assuming that the
graph is a circuit, the path cover is relatively easy to compute by
starting at each input and finding a successor among its neighbours.
If there is ever a choice of successor for a given vertex this choice
can be resolved by looking at the other wires.  This information is
needed frequently by our optimisation procedure.  If an unambiguous
path cover cannot be constructed then the procedure halts with an
error -- in this case the input graph was not a circuit.

Our optimiser can be summarised as the repeated alternation of two
phases:
\begin{itemize}
\item Simplification: apply everywhere possible a strictly reducing
  set of rules.
\item Commutation: using targeting and metric rewrites, apply
  reversible rules to (i) move Pauli operators closer to the inputs.
(ii) move \CNOT and \TONC gates which act on the same two qubits
    closed together.
\end{itemize}
In order to reduce the number of rules, there is an initial phase
which replaces all Hadamard, $Z(-\frac{\pi}{2})$ and
$X(-\frac{\pi}{2})$ gates with $S$, $V$ and Paulis as appropriate.
The single qubit gates are then ``tidied up'' by a final pass.

The simplification phase is straightforward, exploiting the different
cases of \ref{prop:CC2} as the key rules.  The commutation phase is
more delicate, and makes heavy use of the path cover computation.

For the single-qubit commutation rules we use targeting function
which traverses each path starting from the input and applies the rule
at the first possible position in the graph.  This works very well
and, in concert with the simplification rules, is guaranteed to reduce
all single qubit Cliffords to one of the forms of \CC1.  However this
approach fails for rules where a Pauli must commute with a \CNOT gate.
The REWRITE\_TARGETED method allows the user to identify a single
vertex in the rewrite rule with a single vertex in the graph to be
matched.  However this is rarely sufficient to uniquely determine the
overall match.  For example, the following rule
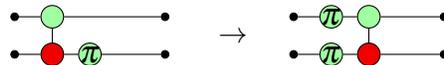
\ctikzfig{derived-rules/GreenPiCx}
can be applied in two different ways to the central diagram below:
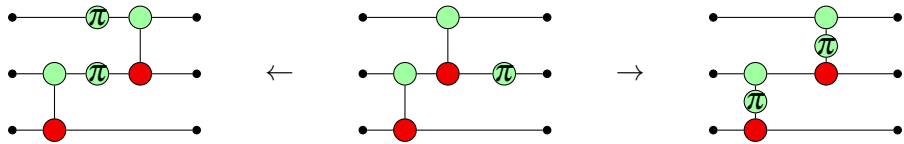
\[
\begin{tikzpicture}[baseline={([yshift=-.5ex]current bounding box.center)}]
	\begin{pgfonlayer}{nodelayer}
		\node [style=Z dot] (v4) at (-4.32, 3.75) {};
		\node [style=wire] (in0) at (-4.88, 4.50) {};
		\node [style=wire] (out0) at (-2.43, 4.50) {};
		\node [style=gpi] (v3) at (-3.75, 3.75) {};
		\node [style=wire] (out2) at (-2.43, 3.00) {};
		\node [style=wire] (in1) at (-4.88, 3.75) {};
		\node [style=Z dot] (v1) at (-3.18, 4.50) {};
		\node [style=X dot] (v0) at (-3.18, 3.75) {};
		\node [style=gpi] (v2) at (-3.75, 4.50) {};
		\node [style=wire] (in2) at (-4.88, 3.00) {};
		\node [style=X dot] (v5) at (-4.32, 3.00) {};
		\node [style=wire] (b3) at (-2.43, 3.75) {};
	\end{pgfonlayer}
	\begin{pgfonlayer}{edgelayer}
		\draw [style=simple] (b3) to (v0);
		\draw [style=simple] (v5) to (in2);
		\draw [style=simple] (v1) to (v0);
		\draw [style=simple] (v0) to (v3);
		\draw [style=simple] (v1) to (v2);
		\draw [style=simple] (v4) to (in1);
		\draw [style=simple] (v5) to (out2);
		\draw [style=simple] (v3) to (v4);
		\draw [style=simple] (v1) to (out0);
		\draw [style=simple] (v2) to (in0);
		\draw [style=simple] (v4) to (v5);
	\end{pgfonlayer}
\end{tikzpicture}

\qquad \leftarrow \qquad 
\begin{tikzpicture}[baseline={([yshift=-.5ex]current bounding box.center)}]
	\begin{pgfonlayer}{nodelayer}
		\node [style=Z dot] (v4) at (-4.32, 3.75) {};
		\node [style=wire] (in0) at (-4.88, 4.50) {};
		\node [style=wire] (out0) at (-2.43, 4.50) {};
		\node [style=wire] (out2) at (-2.43, 3.00) {};
		\node [style=wire] (in1) at (-4.88, 3.75) {};
		\node [style=Z dot] (v1) at (-3.75, 4.50) {};
		\node [style=X dot] (v0) at (-3.75, 3.75) {};
		\node [style=gpi] (v2) at (-3.00, 3.75) {};
		\node [style=wire] (in2) at (-4.88, 3.00) {};
		\node [style=X dot] (v5) at (-4.32, 3.00) {};
		\node [style=wire] (b3) at (-2.43, 3.75) {};
	\end{pgfonlayer}
	\begin{pgfonlayer}{edgelayer}
		\draw [style=simple] (v2) to (b3);
		\draw [style=simple] (v5) to (in2);
		\draw [style=simple] (v0) to (v2);
		\draw [style=simple] (v1) to (v0);
		\draw [style=simple] (v0) to (v4);
		\draw [style=simple] (v4) to (in1);
		\draw [style=simple] (v5) to (out2);
		\draw [style=simple] (v1) to (out0);
		\draw [style=simple] (in0) to (v1);
		\draw [style=simple] (v4) to (v5);
	\end{pgfonlayer}
\end{tikzpicture}

\qquad  \rightarrow\qquad 
\begin{tikzpicture}[baseline={([yshift=-.5ex]current bounding box.center)}]
	\begin{pgfonlayer}{nodelayer}
		\node [style=gpi] (v4) at (-3.38, 4.12) {};
		\node [style=wire] (in0) at (-4.88, 4.50) {};
		\node [style=wire] (out0) at (-2.43, 4.50) {};
		\node [style=Z dot] (v3) at (-4.32, 3.75) {};
		\node [style=wire] (out2) at (-2.43, 3.00) {};
		\node [style=wire] (in1) at (-4.88, 3.75) {};
		\node [style=Z dot] (v1) at (-3.38, 4.50) {};
		\node [style=X dot] (v0) at (-3.38, 3.75) {};
		\node [style=gpi] (v2) at (-4.32, 3.38) {};
		\node [style=wire] (in2) at (-4.88, 3.00) {};
		\node [style=X dot] (v5) at (-4.32, 3.00) {};
		\node [style=wire] (b3) at (-2.43, 3.75) {};
	\end{pgfonlayer}
	\begin{pgfonlayer}{edgelayer}
		\draw [style=simple] (b3) to (v0);
		\draw [style=simple] (v5) to (in2);
		\draw [style=simple] (v2) to (v5);
		\draw [style=simple] (v3) to (v2);
		\draw [style=simple] (v3) to (v0);
		\draw [style=simple] (v0) to (v4);
		\draw [style=simple] (v1) to (v4);
		\draw [style=simple] (v3) to (in1);
		\draw [style=simple] (v5) to (out2);
		\draw [style=simple] (v1) to (out0);
		\draw [style=simple] (in0) to (v1);
	\end{pgfonlayer}
\end{tikzpicture}

\]
Note that the righthand rewrite does not produce a circuit.  The current
implementation provides no direct way for a simproc to specify that
one matching should be preferred to another, even when it has
the necessary information.  For this reason, we have defined a metric
which favours graphs where the Paulis are nearer to the inputs, and
maximally penalises graphs with vertices which are not in any
path.

The approach sketched here is very conservative: the optimiser only
performs rewrites which result in diagrams which meet the (also very
conservative) definition of circuit, even in cases where that diagram
could be later rewritten to circuit.

\section{Summary of results and discussion}
\label{sec:summ-results-disc}

Since we have formalised circuits in the \zxcalculus, our measure of
size is the number of non-trivial vertices in the graph.  This
quantity is the gate count over the generating set $S$, $Z$, $V$, $X$,
and \CNOT, although each \CNOT contributes 2 to the measure.
A less obvious consequence of the \zxcalculus formalism is that the
swap gate is not represented at all, and \CNOT gates may act on any
qubits without penalty.  This may or may not be realistic for a
given quantum architecture, but it is a fundamental feature of the
\zxcalculus, and removing it would require a string diagrammatic
formalism that takes into account the planarity of diagrams.  This
would take us far beyond the scope of this work.

Our optimiser reduces all one- and two-qubit circuits to elements
of \CC1 and \CC2 respectively; however in larger circuits it does not
always find a minimal form.  Notably our ruleset does not
include any reductions for circuits of three or more qubits, so there
are relatively simple \CNOT circuits which cannot be reduced.  
Further, due to the extreme conservativity of our metric and
targeting procedures, the optimiser does not always find valid
rewrites among the ruleset.  
However, on a test
set of randomly generated Clifford circuits it typically produces
significant reductions of circuit size; see Table \ref{tab:results}.

\begin{table}
  \centering
  \begin{tabular}{c|c|c|c|c|c|c|c}
    width & depth & input size & output size 
    & size reduction & proof steps & time $\pm \sigma$ \\\hline
     1 & 20 & 10.8 & 2.4 & 0.26 & 19.6 & 7 $\pm$ 3\\
     2 & 20 & 30.2 & 9.0 & 0.30 & 78.9 & 1021 $\pm$ 680\\
     3 & 20 & 44.2 & 14.7& 0.32 & 148  & 3144 $\pm$ 2189\\
     5 & 20 & 78.5 & 22.6 & 0.28 & 239 & 6679 $\pm$ 3091\\\hline
  \end{tabular}
  \caption{Summary of results on randomly generated circuits.  Random
    circuits were generated with the program \texttt{randomCliffs.py}
    which is available from the project repository.  Tests were
    performed on a basic desktop PC.}
  \label{tab:results}
\end{table}

There is striking increase in compute time per proof step as the width
of the circuit increases.  This can be attributed to the increased
number of \CNOT gates: as the average vertex degree increases, the
number of candidate matches including a given vertex increases
geometrically.  Since many of these matches will be invalid as circuit
rewrites, a lot of time is spent searching for matches which are then
rejected.  Further, wider circuits require two-qubit
commutation rules, which are metric-driven rather than targeted, and
hence impose a much greater computational load.

Indeed, the main difficulty we have faced is taking control of the
matching engine.  Targeted rewriting can select good rewrite positions
with high efficiency but does not discriminate between good and bad
rewrites at that position.  Metric-driven rewrites can perform such
discrimination but at the cost of losing control over the search.
Combining these features in a single combinator would greatly ease
this difficulty.

Further, the targeting and metric functors are required to be
stateless; this imposes another avoidable performance cost, as the
path cover of the graph must be recomputed on every invocation,
regardless of whether a rewrite was performed.  Given the poor
performance of interpreted Python compared to the Scala core,
significant speed-up could be achieved by providing a range of
built-in metrics for use in simprocs.

\section{Conclusions and Future Work}
\label{sec:future-work}

Since Quantomatic was designed for interactive use, its simproc
environment is currently quite limited.  In repurposing the system to
operate as a fully automated circuit optimisation tool, we have run up
against its current limits, and our frequent resort to brute force and
perversity is a reflection of these limitations,
rather than the authors' mental states\footnote{However the first
  author's mental state did take a beating from the lack of debugging
  facilities.}.  These limitations are not fundamental, and
could be removed by minor extensions of the simproc API.  Although
there is much scope for improvement in our optimiser, we have
demonstrated that the Quantomatic system can serve as the basis for an
effective circuit optimisation tool.





\small
\bibliography{all}

\newpage
\normalsize



\appendix

\section{Electronic Resources}
\label{sec:electronic-resources}

All the work described in this paper is available as a downloadable
quantomatic project from the following URL.
\begin{center}
  \projectrepo
\end{center}
Please download it and play around!

\section{An aside on CNOT and SWAP}
\label{sec:an-aside-cnot}

We have already seen the
\zxcalculus representation of $\CNOT$ in
Section~\ref{sec:form-cliff-circ}.  Note 
that the upper ($z$) vertex corresponds to the control qubit, while the
lower ($x$)  is the target qubit.  We could also consider the $\CNOT$
with control and target reversed, which we denote $\TONC$.  
\[
\TONC \quad := \quad 
\toncSWAP \quad = \quad 
\input{zx-definition/TONC.tikz}
\]
It is a well known fact that the swap itself can be represented as
$\CNOT\circ\TONC\circ\CNOT$, and indeed this is provable in the
\zxcalculus.
\begin{equation}\label{eq:swap-as-tonc}
  \cnottonccnotI \eq
\cnottonccnotII \eq
\cnottonccnotIV \eq 
\cnottonccnotV \eq
\SWAP
\end{equation}
Note that the bialgebra rule gives the first equation, and the Hopf
rule the second: we cannot rely purely on the group structure.

Alternatively, by the colour
duality principle we can express $\TONC$ in terms of $\CNOT$ and $H$: 
\begin{equation}\label{eq:tonc-nf}
\input{zx-definition/TONC.tikz} \quad  \eq \quad \toncHHHH \quad \eq \quad \toncNF
\end{equation}
Combining (\ref{eq:tonc-nf}) and (\ref{eq:swap-as-tonc}) we obtain:
\begin{lemma}\label{lem:swap-as-cnot}
  \[
  \SWAP \qquad \eq \qquad \swapNF
  \]
\end{lemma}

Lemma~\ref{lem:swap-as-cnot} indicates that we can dispense with the
swap as well as \TONC and work purely with \CNOT.  Since we have
formalised the system as a PROP it is more convenient to treat
$\sigma$ as one of the generators of \C2.

\newpage
\section{Proofs omitted from main article}
\label{sec:proofs-omitted-from}

\paragraph*{{Lemma \ref{lem:colourchange}}}
  The following rule is admissible:
  \[
  \inline{\input{gen-colour-change-lhs.tikz}} \eq
  \inlinetikzfig{gen-colour-change-rhs}  
  \]
where $n$ is the number of red vertices in the left hand side diagram.

\begin{proof}
  Consider a single red spider.  By the H-rules, and the spider we may
  reason as follows: 
  \begin{align*}
    \input{gen-colour-change-proof-i.tikz} 
    & \quad\rw\quad
    \input{gen-colour-change-proof-ii.tikz} 
    \quad\rw\quad
    \input{gen-colour-change-proof-iii.tikz} \\ \\
    & \quad\rw\quad
    \input{gen-colour-change-proof-iv.tikz} 
    \quad\rw\quad
    \input{gen-colour-change-proof-v.tikz}
  \end{align*}
  From which the result follows by applying $z_-$ to every leg.
\end{proof}

\paragraph*{{Proposition \ref{prop:CC2}}}
    Let $k\in \CC2$ and let $g$ be a generator of $\C2$ as a \zxcalculus
  term; then $g\circ k \eq k'$ for some $k' \in \CC2$.

  Here we fill in the missing details behind the equations (i) -- (iv)
  that appeared in the proof.

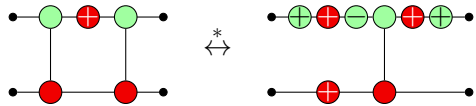
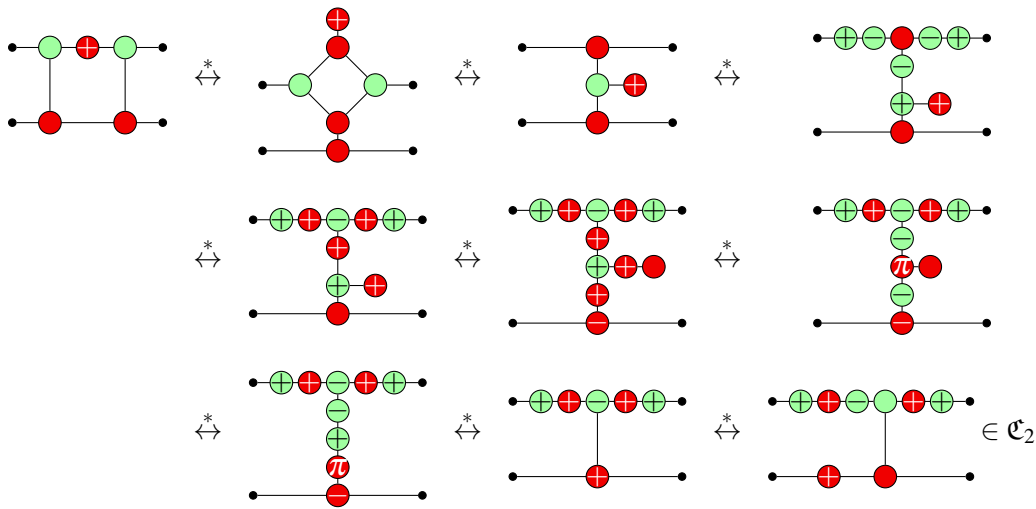
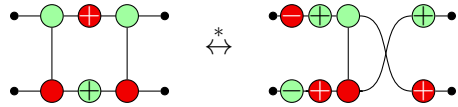
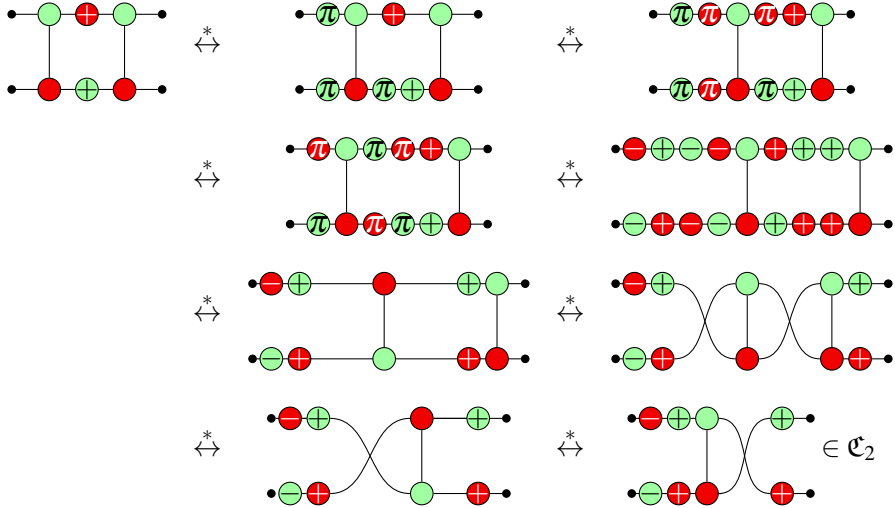
\vstrut{1ex}

\noindent
\emph{Equation (i):}
\[
\CNOTCNOTsCaseAi \quad\eq\quad  \CNOTCNOTsCaseAx
\]
\begin{proof}
\[
\begin{array}{ccccccc}
  \CNOTCNOTsCaseAi 
   & \eq &  \CNOTCNOTsCaseAii 
   & \eq &  \CNOTCNOTsCaseAiii 
   & \eq &  \CNOTCNOTsCaseAiv 
  \\ 
  &\qquad&&\qquad&&\qquad&
  \\
   & \eq &  \CNOTCNOTsCaseAv
   & \eq &  \CNOTCNOTsCaseAvi 
   & \eq &  \CNOTCNOTsCaseAvii
  \\
  \\
   & \eq &  \CNOTCNOTsCaseAviii
   & \eq &  \CNOTCNOTsCaseAix 
   & \eq &  \CNOTCNOTsCaseAx \in \CC2
\end{array}
\]
\end{proof}

\vstrut{1ex}

\noindent
\emph{Equation (ii):}
\[
\CNOTCNOTsCaseBi \quad\eq\quad  \CNOTCNOTsCaseBix
\]
\begin{proof}
\[
\begin{array}{ccccc}
  \CNOTCNOTsCaseBi 
   & \eq &  \CNOTCNOTsCaseBii 
   & \eq &  \CNOTCNOTsCaseBiii 
  \\ 
  &\qquad&&\qquad&
  \\
   & \eq &  \CNOTCNOTsCaseBiv 
   & \eq &  \CNOTCNOTsCaseBv
  \\ 
  \\
   & \eq &  \CNOTCNOTsCaseBvi 
   & \eq &  \CNOTCNOTsCaseBvii
  \\
  \\
   & \eq &  \CNOTCNOTsCaseBviii
   & \eq &  \CNOTCNOTsCaseBix \in \CC2
\end{array}
\]
Note that  we have used both Equations  (\ref{eq:swap-as-tonc}) and
(\ref{eq:tonc-nf}) in the heart of this proof.
\end{proof}

\vstrut{1ex}

\noindent
\emph{Equation (iii):}
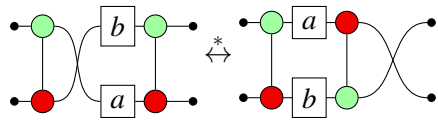
\[
\qquad\CNOTswapLHS \eq \CNOTswapRHS
\]
\begin{proof}
  This is a straightforward application of the ``topology'' principle;
  no rules are required, they are equal in the sense of
  Definition~\ref{def:zxterms}.
\end{proof}

\vstrut{1ex}

\noindent
\emph{Equation (iv):}
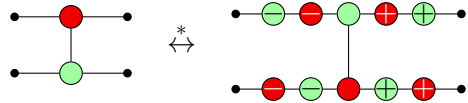
\[
  \input{zx-definition/TONC.tikz}  \quad \eq \quad \toncNF
\]
\begin{proof}
This is simply Equation (\ref{eq:tonc-nf}) from the preceding section.
\end{proof}





\newpage

\section{Main Rules}
\label{sec:main-rules}

The following derived rules are the core of the circuit
transformations in the main proof development.  The rules themselves are
found in \texttt{derived-rules/} in the Quantomatic project, and their
derivations are found in \texttt{derived-rules-proofs/}.


In all diagrams inputs are to the left, outputs to the right.

\subsection{Init}
\label{sec:init}

These rules remove all minus and H nodes from the graph.

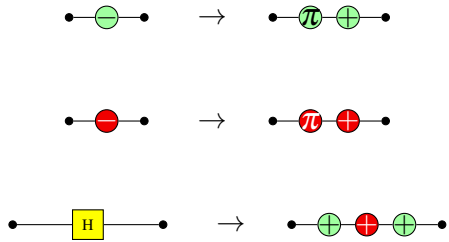
\ctikzfig{derived-rules/GreenMinus}

\ctikzfig{derived-rules/RedMinus}

\ctikzfig{derived-rules/AlwaysH}

\subsection{Always Rules}
\label{sec:always-rules}

These are strictly reducing so can always be applied.

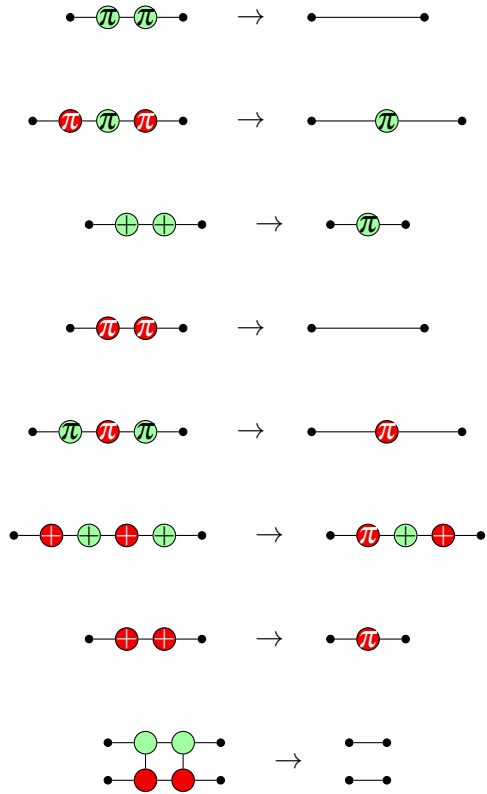
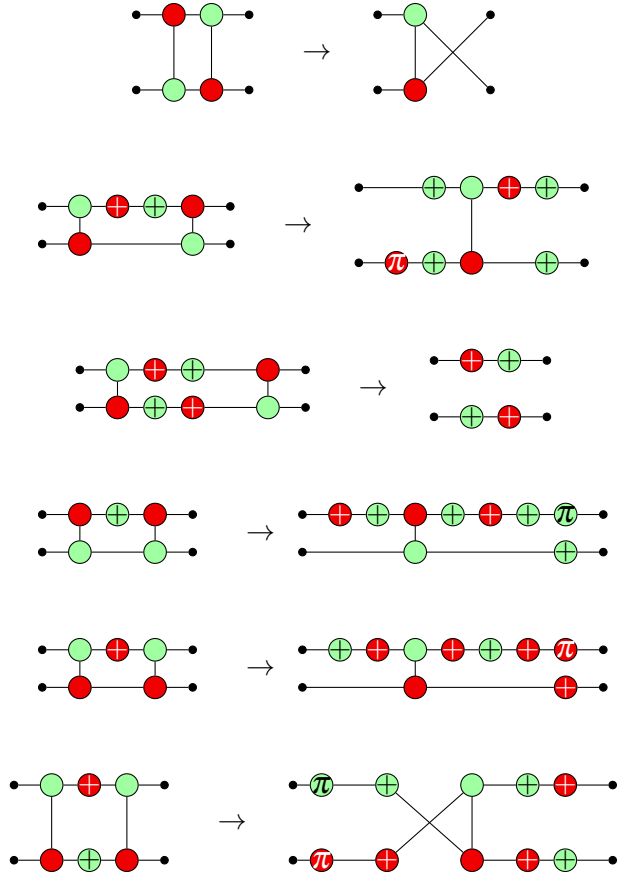
\ctikzfig{derived-rules/GreenPi}

\ctikzfig{derived-rules/GreenPi2}

\ctikzfig{derived-rules/GreenPlus}

\ctikzfig{derived-rules/RedPi}

\ctikzfig{derived-rules/RedPi2}

\ctikzfig{derived-rules/Euler}

\ctikzfig{derived-rules/RedPlus}

\ctikzfig{derived-rules/Cx}

\ctikzfig{derived-rules/CxSw}

\ctikzfig{derived-rules/C22Plus1Bit}

\ctikzfig{derived-rules/C22Plus2Bit}

\ctikzfig{derived-rules/C2GreenPlus1Bit}

\ctikzfig{derived-rules/C2RedPlus1Bit}

\ctikzfig{derived-rules/C2Plus2Bit}

\subsection{EulerProc}

\ctikzfig{derived-rules/H}

The inverse is also used.

\subsection{Pauli Commutation rules}
\label{sec:pauli-comm-rules}

Each of these is managed by a custom simproc which searches for a good
place to apply them.

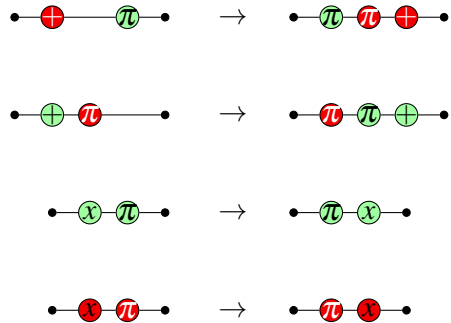
\ctikzfig{derived-rules/GreenPiCommute}

\ctikzfig{derived-rules/RedPiCommute}

\ctikzfig{derived-rules/GreenCommute}

\ctikzfig{derived-rules/RedCommute}

These rules are used by the Pauli Commutation proc, which is based on
REDUCE\_METRIC.

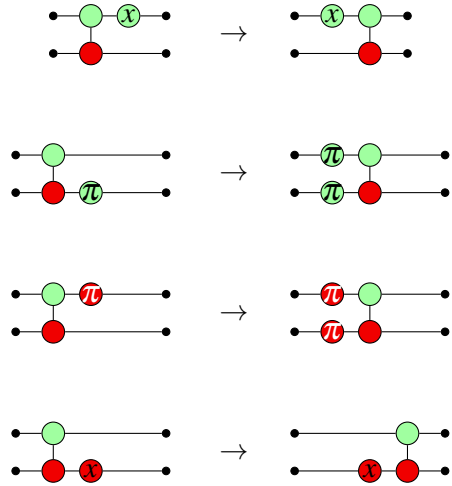
\ctikzfig{derived-rules/GreenCxCommute}

\ctikzfig{derived-rules/GreenPiCx}

\ctikzfig{derived-rules/RedPiCx}

\ctikzfig{derived-rules/RedCxCommute}

\subsubsection{C2 Rules}
\label{sec:c2-rules}

These two alternate in the C2Proc simproc.  Both have a custom
targeting routine to select where they should be applied.

\ctikzfig{derived-rules/C2RedCxCommute}

\ctikzfig{derived-rules/C2GreenCxCommute}

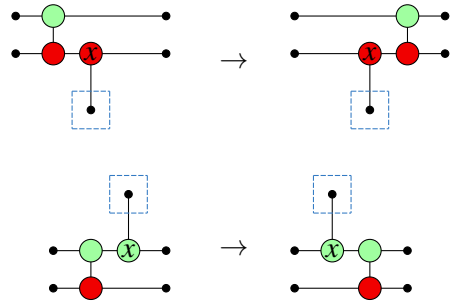
\end{document}